\documentclass[11pt]{article}
\usepackage[letterpaper,textwidth=6.5in,textheight=9in,
            centering,ignorehead,nomarginpar]{geometry}
\usepackage{color}
\usepackage{array}
\usepackage{natbib}
\bibpunct{(}{)}{;}{a}{,}{,}
\usepackage{graphicx}
\usepackage{bbm}
\usepackage{amsmath,amsthm,amssymb,amsfonts,latexsym}
\usepackage{booktabs}
\usepackage{icomma}
\usepackage{afterpage}
\usepackage[labelfont={small,sc},textfont={small,it},
            margin=15pt,skip=10pt,position=bottom]{caption}
\usepackage[labelfont={scriptsize,normalfont},textfont={scriptsize,it},
            margin=20pt,skip=5pt,position=bottom,
            labelformat=simple]{subcaption}

\usepackage{algorithmic}
\usepackage{algorithm}

\newtheorem{remark}{Remark}
\newtheorem{assumption}{Assumption}

\newtheorem{theorem}{Theorem}

\newtheorem{corollary}{Corollary}
\newtheorem{lemma}{Lemma}

\numberwithin{condition}{section}
\numberwithin{assumption}{section}
\numberwithin{remark}{section}
\numberwithin{equation}{section}
\numberwithin{lemma}{section}
\numberwithin{definition}{section}
\numberwithin{theorem}{section}
\numberwithin{proposition}{section}
\numberwithin{table}{section}
\numberwithin{figure}{section}
\numberwithin{theorem}{section}
\numberwithin{corollary}{section}
\numberwithin{property}{section}

\newcommand{\EQ}{\begin{equation}}
\newcommand{\EN}{\end{equation}}
\newcommand{\EQS}{\begin{equation*}}
\newcommand{\ENS}{\end{equation*}}

\def\n1{n}

\newsavebox{\savepar}

\numberwithin{equation}{section}
\numberwithin{table}{section}
\numberwithin{figure}{section}

\usepackage[running]{lineno}
\newcommand*\patchAmsMathEnvironmentForLineno[1]{%
  \expandafter\let\csname old#1\expandafter\endcsname\csname #1\endcsname
  \expandafter\let\csname oldend#1\expandafter\endcsname\csname end#1\endcsname
  \renewenvironment{#1}%
     {\linenomath\csname old#1\endcsname}%
     {\csname oldend#1\endcsname\endlinenomath}}%
\newcommand*\patchBothAmsMathEnvironmentsForLineno[1]{%
  \patchAmsMathEnvironmentForLineno{#1}%
  \patchAmsMathEnvironmentForLineno{#1*}}%
\AtBeginDocument{%
\patchBothAmsMathEnvironmentsForLineno{equation}%
\patchBothAmsMathEnvironmentsForLineno{align}%
\patchBothAmsMathEnvironmentsForLineno{flalign}%
\patchBothAmsMathEnvironmentsForLineno{alignat}%
\patchBothAmsMathEnvironmentsForLineno{gather}%
\patchBothAmsMathEnvironmentsForLineno{multline}%
}

\begin{document}


\title{
     $\epsilon$-Monotone Fourier Methods for Optimal Stochastic Control in Finance
}

\author{Peter A. Forsyth \thanks{David R. Cheriton School of Computer Science,
        University of Waterloo, Waterloo ON, Canada N2L 3G1,
        \texttt{paforsyt@uwaterloo.ca}, +1 519 888 4567 ext.\ 34415.}
        \and
         George Labahn\thanks{David R. Cheriton School of Computer Science,
        University of Waterloo, Waterloo ON, Canada N2L 3G1,
        \texttt{glabahn@uwaterloo.ca}, +1 519 888 4567 ext.\ 34667}
}

\maketitle

\vspace{.2in}
\begin{abstract}
\vspace{10pt}

\vspace{.1in}

\noindent {\color{black}
Stochastic control problems in finance often involve complex controls at 
discrete times. As a result numerically solving such problems, for example using methods based
on partial differential or integro-differential equations, 
inevitably give rise to low order accuracy, usually at most second order. In 
many cases one can make use of Fourier methods to efficiently advance solutions
between control monitoring dates and then apply numerical optimization methods 
across decision times. However Fourier methods are not monotone and 
as a result give rise to possible violations of arbitrage inequalities.  This
is problematic in the context of control problems, where the control is determined by comparing value functions.
In this paper we give a preprocessing step for Fourier
methods which involves projecting the Green's function
onto the set of linear basis functions.  The resulting
algorithm is guaranteed to be monotone (to within a
tolerance), $\ell_\infty$-stable and satisfies an $\epsilon$-discrete comparison principle.
In addition the algorithm has the same complexity per step as
a standard Fourier method while at the same time having second
order accuracy for smooth problems.}

\noindent

\vspace{.1in}


\vspace{.1in}

\noindent
\textbf{Running Title:} $\epsilon$-Monotone Fourier

\vspace{.1in}

\noindent
\textbf{Key Messages:}
\begin{itemize}
    \item Current Fourier methods (FST/CONV/COS) are not necessarily monotone
    \item We devise a pre-processing step for FST/CONV methods which are monotone
          to a user specified tolerance
    \item The resulting methods can be used safely for optimal control problems
          in finance
\end{itemize}

\end{abstract}

\vspace{.1in}

\section{Introduction}

{\color{black}
Optimal stochastic control problems in finance often involve monitoring or making decisions at discrete points in time.  
These monitoring times typically cause difficulties when solving  optimal stochastic control problems numerically, both for efficiency and correctness. 
Efficiency because numerical methods are typically applied from one monitoring time to the next. Correctness arises as an issue when the decision is determined by comparing value functions, something problematic when discrete approximations are not monotone. 
These optimal stochastic problems arise in many 
important financial applications.  This includes problems such as asset allocation  
\citep{li-ng:2000,huang2010,forsyth2016abc,Cong2016},
pricing of variable annuities  
\citep{Bauer2008,Kwok2008,Chen2008,Song2016,Garcia2017,Kwok2018},
and hedging in discrete time  \citep{Rem2013,Ang2014} to name just a few.}

These optimal control problems are typically modeled as the solutions of Partial Integro Differential Equations (PIDEs), 
which can be solved via numerical, finite difference  \citep{Chen2008} or Monte Carlo \citep{Cong2016}  
methods. When cast into dynamic programming form, the optimal control problem reduces to solving a PIDE backwards in time 
between each decision point, and then determining the optimal control at each such point.
In many cases, including for example those mentioned above, the models are based on fairly simple
stochastic processes, with the main interest being the behaviour of the optimal controls. These simple stochastic models can be justified if one is looking at long term problems, for example, variable annuities or saving for retirement, where
the time scales are of the order of 10-30 years. In these situations it is reasonable to use a parsimonious stochastic
process model.

In these (and many other) situations the characteristic function of the associated stochastic process is known in closed form. For the type of PIDEs that appear in financial problems, knowing the characteristic function implies that the Fourier transform of the solution is also known in closed form. By discretizing these Fourier transforms we obtain an
approximation to the solution which can be used for effective numerical computation. 
A natural approach in this case is to use a Fourier scheme to advance the solution in a single time step between decision times, and then apply a numerical optimization approach to advance the solution across the decision time. This technique is repeated until the current time is reached \citep{Ruijter2013,Lippa2013}. These methods are based on Fourier Space Timestepping (FST) \citep{Surkov2008}, the CONV technique \citep{Lord2008} or the COS algorithm \citep{Fang2008}. Fourier methods have been applied to pricing of exotic variance products and volatility derivatives \citep{Zheng2014}, guaranteed minimum withdrawal benefits 
\citep{Song2016,Garcia2017,Kwok2018} 
and equity-indexed annuities \citep{Deng2017} to name just a few.

Fourier methods have a number of advantages compared to finite difference and other methods. First and foremost is that there are no timestepping errors between decision dates. These methods also provide easy handling for stochastic processes involving jump diffusion (\cite{Lippa2013}) and regime switching (\cite{Surkov2008}). 
Although Fourier methods typically need a large number of  discretization points, 
the algorithms reduce to using finite FFTs which are efficiently available on most platforms (including also GPUs). The algorithms are also quite easy to implement. For example, using Fourier methods for the pricing of variable annuities 
reduces to the use of discrete FFTs and local optimization. Detailed knowledge of PDE algorithms is not actually required in this case.
Fourier methods also easily extend to  multi-factor stochastic process where  finite difference methods have difficulties because
of cross derivative terms. 
{\color{black}
Of course, Fourier methods suffer from the curse of dimensionality, and hence
are restricted, except in special cases, to problems of dimension three or less.}
Finally, Fourier methods have good convergence properties for problems
with non-complex controls. For example, for
European option pricing, in cases  where the characteristic function of the underlying stochastic process is known, the COS
method achieves exponential convergence (in terms of number of terms of the Fourier series) \citep{Fang2008}.   

A major drawback with current, existing Fourier methods is that they are not monotone.  In the contingent claims context, monotone methods preserve arbitrage inequalities, 
or discrete comparison properties, independent of any discretization errors.  
As a concrete example, consider the case of a variable annuity contract, with 
ratchet features
and withdrawal controls at each decision date. Suppose contract A has a larger payoff at the terminal time than contract B.  Then a monotone scheme generates a value for contract A which
is always larger than the value of contract B, at all points in time and space, 
regardless of the accuracy of the  numerical scheme. In a sense, the arbitrage 
inequality (discrete comparison) condition is the financial 
equivalence of conservation of mass  in engineering computations. 
Use of non-monotone methods is especially problematic in the context of 
control problems, where the control is determined by comparing value functions. 

Monotonicity is also relevant for the convergence of numerical schemes. In general,
optimal control problems posed as PIDEs  are nonlinear and do not have unique solutions. 
The financially relevant solution is the viscosity solution of the 
PIDEs and it is well known \citep{barles-souganidis:1991} that a discretization of a PIDE converges to the viscosity solution if it is monotone, consistent and stable. There are examples \citep{Oberman2006} where
non-monotone discretizations fail to converge and also 
examples where there is convergence \citep{pooley2002} 
but not to the financially correct viscosity solution. In addition, in cases where the Green's function
has a thin peak, existing non-monotonic Fourier methods require a very small space step which often results in numerical issues.
Finally, monotone schemes are more reliable  for the numeric computation of Greeks (i.e. derivatives of the solution), 
often important information for  financial instruments.

{\color{black}
The starting point for this paper is the assumption that we have a closed form representation of the
Fourier transform of the Green's function of the stochastic process PIDE. From  a practical point of view, 
we also assume that a spatial shift property also
holds.  This last assumption can be removed 
but at a cost of reducing the good computational complexity of our method.
We will discuss these assumptions further in subsequent sections.
}

In this paper we present a new Fourier algorithm in which monotonicity can be guaranteed  to within a user specified numerical tolerance. The algorithm is for use with general optimal control problems in finance. In these  general control problems the
objective function may be complex and non-smooth, and hence the optimal control at each step must  be determined by a
numerical optimization procedure. Indeed, in many cases, this is done by discretizing the control and using exhaustive search. Reconstructing the Fourier coefficients is typically done by assuming the control is constant over discretized intervals of the physical space, numerically determining the control at the midpoint of these intervals and finally by reconstructing the Fourier coefficients by quadrature. This is equivalent to using a type of trapezoidal rule to reconstruct the Fourier coefficients and hence this can be at most second order accurate (in terms of the physical domain mesh size).

In fact we show how one can modify the FST or CONV schemes to get new schemes in which monotonicity can be guaranteed  to within a user specified numerical tolerance. Our approach is similar to that used in these schemes which first approximate the solution of a linear PIDE by a Green's function convolution, then discretize the convolution and finally carry out the dense matrix-vector multiply efficiently using an FFT. In our case we discretize the value function, and generate a continuous approximation of the function by assuming linear basis (or alternatively piecewise constant)  functions.  Given this approximation, we carry out an exact integration of the convolution integral and then truncate the series approximation of this integral so that monotonicity holds to within a tolerance.  
Consequently, we prove that our algorithm has an 
$\epsilon-$Discrete Comparison Property, that is,  
given a tolerance $\epsilon$, then a 
discrete comparison (a.k.a. arbitrage inequality) holds  
to $O(\epsilon)$, independent of the number of discretization nodes and timesteps.
This is similar in spirit to the $\epsilon$-monotone schemes 
discussed in, for example, \citet{Bokanowski}. 
Typically the convergence to the integral is exponential in the series truncation parameter so it is inexpensive to specify a small tolerance. The key idea here is that the number of terms required to accurately determine the projection of the Green's function onto linear basis functions can be larger than the number of basis functions. 
After an initial set up cost,
the complexity per step is the same as the standard FST or CONV methods.
This requires only a small change to existing codes in order to guarantee monotonicity. The desirable property of our method is that monotonicity can be guaranteed (to within a tolerance) independent of the number of FST (CONV) grid 
nodes or time-step size.

While Fourier methods have good convergence properties for  vanilla
contracts or problems where controls are smooth, it is a
different story for general optimal control problems. For example, 
if the COS method is applied to optimal control problems,
then it is challenging to maintain exponential convergence as the 
optimal control must be determined in the physical space.
Hence, a highly accurate recursive expression for the Fourier coefficients must be found after application of the optimal
controls, in order to maintain exponential convergence. In the case of 
bang-bang controls, it is often possible to separate the
physical domain into regions where the control is constant. 
If these regions are determined to high accuracy, then an accurate
algorithm for recursive generation of the Fourier coefficients 
can be developed (\cite{Ruijter2013}). However, even for the case
of an American option, this requires careful analysis and implementation (\cite{Fang2008}). Our interest is in general problems,
where the control may not be of the bang-bang type, and we expect that such good convergence properties will not hold.  
In addition, in the path dependent case, the problem is usually converted into Markovian form through
additional state variables.  The dynamics of these state variables are typically represented by a
deterministic equation (between monitoring dates). At monitoring dates, the state variable may have non-smooth
jumps (e.g. cash flows) and hence the standard approach would be to discretize this state variable, and then
to interpolate the value function across the monitoring dates.  If linear interpolation is used, this also
implies that the solution is at most second order accurate at a monitoring date.

While monotone schemes have good numerical properties they appear to inherently be low order methods.  However,  it would seem that in the most general case, it is difficult to develop high order schemes for control problems.  For example, in the COS method, this problem can be traced to the difficulty of reconstructing the Fourier coefficients after numerically determining the optimal control at discrete points in the physical space. Consequently, in this article we focus on FST or CONV techniques, which use straightforward procedures to move between Fourier space and the physical space (and vice versa).  

We illustrate the behaviour of our algorithm by comparing various implementations of FST/CONV on some model option pricing
examples, in particular European and Bermudan options.  In addition, we demonstrate the use of the monotone scheme
methods on a realistic asset allocation problem. Our main conclusion is that for problems with complex controls, where we can
expect fairly low order convergence to the solution, a small change to standard FST or CONV methods can be made which
guarantees monotonicity, at least to within a user specified tolerance. This does not alter the order of convergence in this case,
hence we can ensure a monotone scheme  with only a slightly increased set up cost. After the initialization, the complexity per
step of the monotone method is the  same as the standard FST/CONV algorithm.

The remainder of this paper is as follows. In the next section we describe our optimal control problem in a general setting. Section \ref{sec3} is used to describe existing Fourier methods which allows us  to contrast with our new monotone Fourier method presented in Section \ref{sec4}. The monotone algorithm for solving optimal control problems is then given in Section \ref{sec5} with properties of the algorithm and proofs appearing in the following section. Wrap-around is an important issue for Fourier methods, particularly in the case of our control problems.  Our method of minimizing such error is described in Section \ref{wrap_section}. Section \ref{sec7} presents two numerical examples used to stress test the monotone algorithm. This is followed by an application of our algorithm to the multiperiod mean variance optimal asset allocation problem, a general optimal control problem well suited to our monotone methods. The paper ends with a conclusion and topics for future research.

\section{General Control Formulation}\label{sec2}

In this section we describe our optimal control problem in a general setting.
Consider a set of intervention (or monitoring) times $t_n$
\begin{equation}
\hat{\mathcal{T}} \equiv \{t_0 \leq \dots  \leq t_M\}
\label{eq:intervention_time}
\end{equation}
with $t_0 = 0$ the inception time of the investment and $t_M=T$ the terminal time. 
For simplicity, we specify the set of intervention times
to be equidistant, that is, 
$t_{n} - t_{n-1} = \Delta t = T/M$ for each $n$.

Let $t_n^- = t_n - \epsilon$ and $t_n^+ = t_n + \epsilon$, with
$\epsilon \rightarrow 0^+$, denote the instant before and after the $n^{th}$
monitoring time $t_n$.  We define a value function $\hat{v}(x,t)$ with domain $x \in \mathbb{R}$ (we
restrict attention to one dimensional problems for ease of
exposition), which satisfies 
\begin{eqnarray}
  \hat{v}_t + \mathcal{L} \hat{v} = 0 &;& t \in (t_n^+, t_{n+1}^-)~
\end{eqnarray}
with $\mathcal{L}$ a partial integro-differential operator.
At $t_n \in \hat{\mathcal{T}}$ we find an optimal control $\hat{c}(x,t_n)$ via
\begin{eqnarray}
   \hat{v}( x, t_n^{-}) &=& \inf_{\hat{c} \in \mathbb{Z} } \mathcal{M}({\hat{c}} ) \hat{v}(x, t_n^+)
\end{eqnarray}
where $\mathcal{M}( {\hat{c}} )$ is an intervention operator and
$\mathbb{Z}$ is the set of admissible controls.  

It is more natural to rewrite these equations going backwards in time $\tau = T - t$, that is, in terms of
time to completion. In this case the value function is $ v(x, \tau) = \hat{v}( x, T -t)$  and satisfies
\begin{eqnarray}
  v_\tau - \mathcal{L} v = 0 ~~~~~~~~~~~~~~~~~~~~& ; & \tau \in ( \tau_n^+, \tau_{n+1}^- )~ \label{PIDE_1} \\
    v( x, \tau_{n}^{+}) = \inf_c \mathcal{M}(c) v(x, \tau_{n}^-) & ; &  \tau_n \in \mathcal{T} ~~~~~~~~~~~~~ .
                                                 \label{intervention_1}
\end{eqnarray}
Here the control  $c( x, \tau) = \hat{c}( x, T- \tau)$ and $\mathcal{T}$ now refers to the
set of backwards intervention times 
\begin{eqnarray}
{\mathcal{T}} &\equiv& \{ \tau_0 \leq \dots  \leq \tau_M \} ~~~ \mbox{ with }  \tau_0 = 0,~~ \tau_M=T  \mbox{ and } \tau_n = T - t_{M-n}~.
                        \nonumber      
\label{eq:intervention_time_back}
\end{eqnarray}

A typical intervention operator has the form
\begin{eqnarray}
    \mathcal{M}({{c}} ) {v}(x, \tau_n^-) 
   &=& {v}(x + \Gamma( x, \tau_n^-, c), \tau_n^- )  ~.
    \label{intervention_def}
\end{eqnarray}
As an example, in the context of portfolio allocation,
we can interpret $\Gamma( x, \tau_i^-, c)$ as a rebalancing rule.
In general, there can also be cash flows associated with the
decision process, as in the case of variable annuities. 
However,  for simplicity we will ignore such a generalization in this paper,
and assume that the intervention operator has the form (\ref{intervention_def}).  
In our asset allocation example (described later), the cash flows are modeled by updating a path dependent variable.

\section{CONV and FST Methods}\label{FST_section}\label{sec3}

In this section, we derive the FST and closely related
CONV technique in an intuitive fashion.  This will allow
us to contrast these methods with the monotone technique
developed in the following section.
For ease of exposition, we will continue to restrict attention to one dimensional
problems. However, there is no difficulty generalizing this
approach to the multi-dimensional case.
In a financial context, we often
have that the variable $x = \log(S) \in (-\infty, \infty)$, 
where $S$ is an asset price.

\subsection{Green's Functions}

A solution of the PIDE (\ref{PIDE_1})
\begin{eqnarray}
   v_{\tau} - \mathcal{L} v = 0 ~;~ \tau \in ( \tau_n, \tau_{n+1} ]~ \nonumber
\end{eqnarray}
can be represented in terms of the Green's function of the PIDE, a function 
typically of the form $g = g( x, x^{\prime}, \Delta \tau)$. However,
in many cases this Green's function will have the form 
$g = g( x - x^{\prime}, \Delta \tau)$ and we will assume this to hold in 
our problems.
{\color{black}
More formally, we make the following assumptions, which we
assume to hold in the rest of this work.
\begin{assumption}[Form of Green's function]\label{green_assump_a}
We assume that the Green's function can be written as
\begin{eqnarray}
   g(x, x^{\prime}, \Delta \tau) & = & g( x - x^{\prime}, \Delta \tau) \nonumber \\
                 & = &  \int_{-\infty}^{\infty} G(\omega, \Delta \tau) e^{  2 \pi i \omega (x - x^{\prime})} ~d\omega
\end{eqnarray}
where $G(\omega, \Delta \tau)$ is known in closed form, and $G(\omega, \Delta \tau)$ is independent
of $(x, x^{\prime})$.
\end{assumption}
\begin{remark}[Assumption \ref{green_assump_a}]
If we view the Green's function as a scaled conditional  probability density $f$ then our assumption is that $f(y|x)$   only depends on $x$ and $y$ via  their difference $f(y|x)~=~f(y -  x)$. This assumption holds for L\'evy processes (independent and stationary increments), but does not hold, for example, for a Heston stochastic
volatility model 
nor for  mean reverting Ornstein-Uhlenbeck processes (but see \citet{surkov_2010,Zhang2012,Shan_2014} for possible work-arounds).
The $\epsilon$-monotonicity modifications described in this paper also hold when we do not have $g(x, x^{\prime}, \Delta \tau) =  g( x - x^{\prime}, \Delta \tau)$ but at the price of reduced efficiency. This is discussed later in Section 4.2. 
The second assumption, that we know the Fourier transform of our Green's function in closed form,  is the case, for example, in situations where the characteristic function of the underlying stochastic process is known. In the case of L\'evy processes, the L\'evy-Khintchine formula provides such an explicit representation of the characteristic function. 
\end{remark}
}

\noindent 
From Assumption \ref{green_assump_a}, the exact solution of our PIDE is then 
\begin{eqnarray}
   v(x, \tau + \Delta \tau) & = & \int_{\mathbb{R}}  g(x - x^{\prime}, \Delta \tau) v( x^{\prime}, \tau)~d x^{\prime}
     ~. \label{green_1}
\end{eqnarray}
The Green's function has a number of important properties \citep{menaldi_1992}. For this work the two properties
\begin{eqnarray}
     \int_{\mathbb{R}} g(x, \Delta \tau)~dx  & = & C_1  ~ \leq  ~1  ~~\mbox{ and } ~~
     g(x, \Delta \tau)   \geq   0  \label{green2} 
    \end{eqnarray}
are particularly important.\footnote{
For the examples in this paper the constant $C_1$ is explicitly given (in each example) in 
Appendix \ref{green_appendix}} These properties are formally proven in \citep{menaldi_1992}, 
but can also be deduced from the interpretation of the Green's function as a scaled probability density.

We define the Fourier  transform pair for the Green's function as
\begin{eqnarray}
    G(\omega, \Delta \tau) & = & \int_{-\infty}^{\infty} g(x, \Delta \tau) e^{ - 2 \pi i \omega x} ~dx
         \nonumber \\
    g(x, \Delta \tau) & = & \int_{-\infty}^{\infty} G(\omega, \Delta \tau) e^{  2 \pi i \omega x} ~d\omega
      ~
         \label{FT_b}
\end{eqnarray}
with a closed form expression for $G(\omega, \Delta \tau)$ being available.

As is typically the case, we assume that the Green's function $g(x, \Delta \tau)$
decays to zero as $ |x|\rightarrow \infty$, that is, 
$g(x, \Delta \tau)$ is negligible outside a region $ x \in [-A ,   A]$. 
Choosing  $x_{\min} < -A$ and $x_{\max} >  A$, we  localize the computational 
domain for the integral in equation (\ref{green_1}) so that $ x \in [x_{\min}, x_{\max}]$.
We can therefore replace the Fourier transform pair (\ref{FT_b}) by their Fourier series equivalent
\begin{eqnarray}
   G(\omega_k, \Delta \tau) & \simeq & \int_{x_{\min}}^{x_{\max}} g( x, \Delta \tau) ~e^{ - 2 \pi i \omega_k x } ~dx
                                     \nonumber \\
   \hat{g}(x, \Delta \tau)  & = &  
      \frac{1}{P} \sum_{k=-\infty}^{\infty}  G(\omega_k, \Delta \tau) ~e^{ 2 \pi i \omega_k x} 
                     ~ \label{series_1} 
\end{eqnarray}
with $ P = x_{\max} - x_{\min}$ and $ \omega_k = \frac{k}{P}$.  
Here the scaling factors in equation (\ref{series_1}) are selected to be consistent
with the scaling in (\ref{FT_b}).
The solution of the PIDE (\ref{green_1}) is then approximated as
\begin{eqnarray}
   v(x, \tau + \Delta \tau) & \simeq & \int_{ x_{\min}}^{x_{\max}}  
         \hat{g}(x - x^{\prime}, \Delta \tau) v( x^{\prime}, \tau)~d x^{\prime}
     . \label{green_2}
\end{eqnarray}
Note that the Fourier series  (\ref{series_1}) implies a periodic extension of $\hat{g}$, that is, 
  $\hat{g}( x + P, \tau) = g(x, \tau).$
The localization assumption also then implies that $v(x, \tau)$ is periodically extended. 

Substituting the Fourier series  (\ref{series_1})  into (\ref{green_2}) gives our approximate solution as
\begin{eqnarray}
     v( x, \tau + \Delta \tau) & \simeq &  
             \frac{1}{P} \sum_{k=-\infty}^{\infty} G( \omega_k, \Delta \tau) e^{ 2 \pi i \omega_k x}
                                         \int_{x_{\min}}^{x_{\max}} v(x^{\prime}, \tau) 
                                                 e^{-2 \pi i \omega_k x^{\prime}} ~dx^{\prime}.
              \label{conv_1}
\end{eqnarray}
Let $ \Delta x = \frac{P}{N}$ and choose points $\{x_j\}, \{x^{\prime}_j\}$ by
\begin{eqnarray}
     & & x_j =  \hat{x}_0 + j \Delta x ~;~  x^{\prime}_j = \hat{x}_0 + j \Delta x ~~ \mbox{ for } j= -N/2, \ldots N/2 -1. 
              \nonumber 
\end{eqnarray}
Then the  integral in (\ref{conv_1}) can be approximated by a quadrature rule with weights $w_\ell$ giving
\begin{eqnarray}
    \int_{x_{\min}}^{x_{\max}} v(x^{\prime}, \tau )  e^{-2 \pi i \omega_k x^{\prime}} ~dx^{\prime} 
           &\simeq&
               \sum_{\ell = -N/2}^{N/2-1}  
                       w_\ell ~ v( x^{\prime}_\ell, \tau) e^{- 2 \pi i \frac{k}{P} x^{\prime}_\ell }  \Delta x \nonumber \\
            & = & P   e^{- 2 \pi i \frac{k}{P} \hat{x}_0 } V( \omega_k, \tau) 
                \label{conv_3}
\end{eqnarray}
where 
\begin{eqnarray}
      V( \omega_k, \tau) & = & \frac{1}{N} \sum_{\ell = -N/2}^{N/2-1} 
       w_\ell ~ v( x^{\prime}_\ell, \tau) e^{- 2 \pi i k \ell/N}
        \label{conv_4}
\end{eqnarray}
is the DFT of  $\{ w_j ~v(x^{\prime}_j, \tau )\}$. 
{\color{black}
I the following,
we will consider two cases for the weights $w_\ell$:
the trapezoidal rule and Simpson's quadrature.}
Using equations (\ref{conv_3}) and (\ref{conv_4}) in equation (\ref{conv_1}),
and truncating the infinite sum to $k \in[ -N/2, N/2-1]$ then gives
{\color{black}
\begin{eqnarray}
    v(x_j, \tau + \Delta \tau) & \simeq & \frac{1}{P} \sum_{k=-N/2}^{N/2-1} e^{ 2 \pi i \frac{k}{P} \hat{x}_0}
                               G( \omega_k, \Delta \tau) e^{ 2 \pi i k j/N} P 
                              e^{- 2 \pi i \frac{k}{P} \hat{x}_0 } 
                               V( \omega_k, \tau)
                           \nonumber \\
              & = & \sum_{k=-N/2}^{N/2-1} G( \omega_k, \Delta \tau) V( \omega_k, \tau) e^{ 2 \pi i k j/N} .
                    \label{convolution_fourier}
\end{eqnarray}
}
Thus 
$\{v(x_j, \tau + \Delta \tau)\}$ 
is the inverse DFT of the product $\{G( \omega_k, \Delta \tau) \cdot V( \omega_k, \tau)\}$.

In summary, one can obtain a discrete set of values for the solution $v$ by first going to the Fourier domain by constructing its Fourier transform $V$ using a set of
quadrature weights and then returning to the physical domain by convolution of $V$ with the Fourier transform of the Green's function. The cost is then the cost of doing
a single FFT and iFFT.

There are four significant approximations in these steps. These include localization of the computational domain,
representation of the Green's function by a truncated Fourier series, a periodic extension of the solution and,
finally, approximation of the integral in equation (\ref{conv_1}) by a quadrature rule. The effect of the errors from 
these approximations has been  previously discussed in \citet{Lord2008} and we refer the reader there for details. 
 
\subsection{The FST/CONV Algorithms}

The FST and CONV algorithms are described using the previous approximations. 
Let $(v^{n})^+$ be the vector of solution values 
just after $\tau_n$ and $w_{quad}$ be the vector of quadrature weights 
\begin{eqnarray}
    (v^{n})^+ &=& [ v(x_{-N/2}, \tau_n^+), \ldots, v(x_{N/2-1}, \tau_n^+) ]~
                  \nonumber \\
     w_{quad} & = & [ w(x_{-N/2}), \ldots, w(x_{N/2-1}) ]~.
              \label{v_vec_def}
\end{eqnarray}
Furthermore let $\mathbb{I}_{\Delta x}(x)$, with  $x_k \leq x \leq x_{k+1} $, be a linear interpolation operator 
\begin{eqnarray}
  \mathbb{I}_{\Delta x}(x)  (v^{n})^+ & = & \theta \cdot v(x_k,  \tau_n^+) ~~+~~
                                            (1-\theta)\cdot v(x_{k+1},  \tau_n^+) 
                                       ~~\mbox{ with } ~~
                                          \theta = \frac{(x_{k+1} - x)}{\Delta x} .
                        \label{interp_def}
\end{eqnarray}
The full FST/CONV algorithm applied to a control problem
is illustrated in Algorithm \ref{alg_fst}.  We refer the reader to 
\citet{Lippa2013,Song2016,Kwok2018}
for applications in finance.

\begin{algorithm}[!h]
\caption{
\label{alg_fst}
FST/CONV Fourier method. $x \circ y$ is the Hadamard product of vectors $x,y$.
}
 \begin{algorithmic}[1]
    \REQUIRE ${G} = \{{G}( \omega_j, \Delta \tau) \}, ~ j=-N/2, \ldots, N/2-1$
    \STATE Input: number of timesteps $M$ and initial solution $(v^0)^-$
    \STATE $(v^0)^+ = \inf_c \mathcal{M}(c) \bigl( ~\mathbb{I}_{\Delta x} (x) (v^0)^- ~ \bigr)$
    \FOR[ Timestep loop]{$m=1,\ldots, M$}

    \STATE ${V}^{m-1} = FFT[~ w_{quad}  \circ (v^{m-1})^+ ~]~~~~~~~~~~~~~~~~~~~~~~~~~~~~~~~~~~~~~~~~~~~~~~~~$ \COMMENT{Frequency domain}
             \label{DFT_v_fst}
    \STATE $(v^{m})^- = iFFT[~ {V}^{m-1} \circ {G} ~]~~~~~~~~~~~~~~~~~~~~~~~~~~~~~~~~~~~~~~~~~~~~~~~~~~~~~~~~~~~$ \COMMENT{Physical domain}
                \label{IDFT_v_fst}
    \STATE $ v( x_j, \tau_{m}^+) = \inf_c \mathcal{M}(c)
              \bigl( ~\mathbb{I}_{\Delta x} (x_j) (v^{m})^- ~\bigr) ~;~ j=-\frac{N}{2}, \ldots, j=\frac{N}{2}-1 ~~~~~~~~~~$ \COMMENT{Optimal control}
                             \label{optimal_control_algo_fst}

   \ENDFOR
 \end{algorithmic}
\end{algorithm}

\begin{remark}
In \citep{Surkov2008}, the authors describe their FST method in slightly different terms. There they use a continuous Fourier transform to convert the PIDE into Fourier space.
The PIDE in physical space then reduces to a linear first-order differential equation in Fourier space which can be solved in closed-form as long as the characteristic function of the associated stochastic process is known in closed form (See Appendix A). In this way, the method is able to produce exact pricing results between monitoring dates (if any) of an option, using a continuous domain. In practice, using a discrete computational domain leads to approximations as a discrete Fourier transform is used to approximate the continuous Fourier transform.
\end{remark}

\section{An $\epsilon$-Monotone Fourier Method}\label{sec4}

Our monotone Fourier method proceeds in a similar fashion as in the previous 
section, but  is based on a slightly 
different philosophy.  We begin by discretizing the value function, 
and then generate a continuous approximation of the value function 
by assuming linear basis functions. 
Given this approximation,  we carry out an exact integration of the convolution integral. We can then truncate the series
approximation  of this integral so that monotonicity holds to within a tolerance (using a truncation parameter to keep track 
of the number of terms). Typically the convergence to the integral is  exponential in  the series truncation 
parameter, so it is inexpensive to specify a small tolerance. The key idea here is that  the number of terms required to accurately determine the projection of the Green's 
function onto a given set of linear basis  functions can be larger than the number of basis functions.

An additional important point is that, after the initial set-up cost, the complexity per  step is the same as the 
standard FST or CONV methods. This requires only a small change to existing FST or CONV 
codes in order to guarantee monotonicity. The desirable property of this method is that monotonicity can be guaranteed 
(to within a small tolerance) independent of the number of FST (CONV) grid nodes or timestep size.

\subsection{A Monotone Scheme}

We proceed as follows. As before we assume a localized computational domain
\begin{eqnarray}
 v(x, \tau + \Delta \tau) & = & \int_{x_{\min}}^{x_{\max}} 
           g( x - x^{\prime}, \Delta \tau) v( x^{\prime}, \tau)~dx^{\prime} ~
                 \label{mono_1}
\end{eqnarray}
and discretize this problem on the grid $\{x_j\}, \{x^{\prime}_j\}$ 
{\color{black}
\begin{eqnarray}
    x_j & = &   \hat{x}_0 + j \Delta x ~;~~~  x^{\prime}_j = \hat{x}_0 + j \Delta x ~;~~~ j= -\frac{N}{2}, \ldots ,\frac{N}{2} -1 \nonumber 
\end{eqnarray} }
where $P = x_{\max} - x_{\min}$ and $ \Delta x = \frac{P}{N}$ with $x_{\min} 
        =  \hat{x}_0  - \frac{N \Delta x}{2}$ and $x_{\max} =  \hat{x}_0  + \frac{N \Delta x}{2}$. 
Setting $v_j(\tau) = v(x_j, \tau)$ we can now represent the solution as a linear combination 
\begin{eqnarray}
   v(x, \tau) &\simeq & \sum_{j=-N/2}^{N/2 -1} \phi_j (x) ~v(x_j, \tau) 
                =   \sum_{j=-N/2}^{N/2 -1} \phi_j (x) ~v_j(\tau)~.
                \label{mono_3}
\end{eqnarray}
{\color{black}
where the $\phi_j$ are piecewise linear basis functions, that is, 
\begin{eqnarray}
      \phi_j(x) & = & \begin{cases}
                         \frac{(x_{j+1} - x)}{\Delta x} & x_j \leq x \leq x_{j+1} \\
                         \frac{(x -x_{j-1}) }{ \Delta x} & x_{j-1} \leq x \leq {x_j} \\
                         0 & {\mbox{ otherwise}}~~~~~.
                      \end{cases} 
     \label{basisfunctions}
\end{eqnarray}
}
Substituting representation (\ref{mono_3}) into equation (\ref{mono_1}) gives
\begin{eqnarray}
   v_k (\tau + \Delta \tau)  & = & \int_{x_{\min}}^{x_{\max}} 
           g( x_k - x, \Delta \tau) ~v( x, \tau)~dx \nonumber \\
          & = & \sum_{j= -N/2}^{N/2-1} v_j(\tau)  \int_{x_{\min}}^{x_{\max}} \phi_j(x) ~
                          g( x_k - x, \Delta \tau) ~dx \nonumber \\
             & = & \sum_{j= -N/2}^{N/2-1} v_j(\tau) ~\widetilde{g} ( x_k - x_j, \Delta \tau) \Delta x
          \label{mono_conv} ~,
\end{eqnarray}
where 
\begin{eqnarray}
                         \nonumber \\
             \widetilde{g} ( x_k - x_j, \Delta \tau) 
          & = &  \frac{1}{\Delta x}  \int_{x_k - x_j - \Delta x}^{ x_k - x_j + \Delta x} 
                                            \phi_{k-j}( x ) ~ g( x, \Delta \tau) ~dx
                   ~. \label{mono_4}
\end{eqnarray}
Here we have used the fact that $\phi_j( x_k - x) =\phi_{k-j}( x )$, a property  
which follows from the properties of linear basis functions. Setting 
$
\ell~=~k~-~j$, $ y_\ell~=~x_k - x_j = \ell \Delta x$ for $ \ell = -\frac{N}{2}, \ldots ,\frac{N}{2} -1 
$
gives
\begin{eqnarray}
   \widetilde{g}( y_\ell, \Delta \tau) & = & \frac{1}{\Delta x} \int_{y_\ell - \Delta x}^{y_\ell  + \Delta x} 
                                  \phi_\ell(  x) ~ g(  x, \Delta \tau) ~dx
               ~ \label{mono_5}
\end{eqnarray}
as the averaged projection
of the Green's function onto the basis functions $\phi_\ell$.
Note that for this projection $\widetilde{g}( y_\ell, \Delta \tau) \geq 0$ 
since the exact Green's function has $g(x) \geq 0$ for all $x$, and of course $\phi_\ell(x) \geq 0$.   
Therefore the scheme  (\ref{mono_conv}) is monotone
for any $N$.

{\color{black}
\begin{remark}[Green's function available in closed form]
If the Green's function is available in closed form, rather
than just its Fourier Transform, then equation (\ref{mono_5}) can
be used to directly compute the $\widetilde{g}( y_\ell, \Delta \tau)$ terms,
as, for example, in \citet{Antti2003}.  However, in general
this will require a numerical integration.  If the Fourier
transform of the Green's function is known, we will 
derive a technique to efficiently compute $\widetilde{g}( y_\ell, \Delta \tau)$
to an arbitrary level of accuracy.
\end{remark}}

\subsection{Approximating the Monotone Scheme}

The scheme (\ref{mono_conv}) is monotone since the weights $\widetilde{g}( y_\ell, \Delta \tau)$ given in (\ref{mono_5}) are nonnegative. However it is only possible for us to approximate these weights and this prevents us from guaranteeing monotonicity. In this subsection we show how we overcome this issue. 

Recall that our starting point is that $G$, the Fourier series of the Green's function, 
is known in closed form. We have then replaced our Green's function 
$g(x, \Delta \tau)$ by its localized, periodic approximation
\begin{eqnarray}
\hat{g}(x, \Delta \tau)  & = &  
      \frac{1}{P} \sum_{k=-\infty}^{\infty}  
             e^{ 2 \pi i \omega_k x} G(\omega_k, \Delta \tau)~~\mbox{ where }~ \omega_k = \frac{k}{P}~\mbox{ and }  ~
   ~ P = x_{\max} - x_{\min}
    \nonumber  \label{finite_g} 
\end{eqnarray}
and then projected   the Green's function onto the linear basis functions.
{\color{black}
Replacing $g(x, \Delta \tau)$ by $\hat{g}(x, \Delta \tau)$ in equation (\ref{mono_5}),
and assuming uniform convergence of the Fourier series (see Appendix),
we integrate equation (\ref{mono_5}) term by term resulting in
\begin{eqnarray}
\widetilde{g}( y_j, \Delta \tau) & = & 
                                           \frac{1}{P} \displaystyle \sum_{k=-\infty}^{\infty}
                                          \biggl( 
                                             \frac{1}{\Delta x} \int_{y_j - \Delta x}^{y_j  + \Delta x}
                                                 e^{ 2 \pi i \omega_k x} \phi_j(x) ~dx                                      
                                           \biggr) 
                                          G(\omega_k, \Delta \tau) . \label{integral}
\end{eqnarray}
In the case of linear basis functions (\ref{basisfunctions})
we convert the complex exponential in (\ref{integral}) into trigonometric functions with the resulting integration giving\footnote{For $\omega_k = 0$, we take the limit $\omega_k \rightarrow 0$.}
\begin{eqnarray}
   \widetilde{g}( y_j, \Delta \tau) & = & 
                                           \frac{1}{P} \displaystyle \sum_{k=-\infty}^{\infty}
                                                     e^{ 2 \pi i \omega_k  y_j} 
                                          \biggl( \frac{ \sin^2 \pi \omega_k \Delta x}
                                                         { ( \pi \omega_k \Delta x)^2}
                                           \biggr) G(\omega_k, \Delta \tau) .
           \label{g_tilde}
\end{eqnarray}}
This is then approximated by truncating the series.  

A key point is that  the truncation of 
the projection of the Green's function 
does not have to use the same number of terms as the number
of basis functions.
That is, set $N^{\prime} = \alpha N$, with $N$ defined in equation (\ref{mono_3}) and 
$ \alpha  = 2^k$ for $ k=1, 2, \ldots$. Suppose we now truncate the 
Fourier series for the projected linear basis form for $\widetilde{g}$ 
to $N^{\prime}$ terms. Let $\widetilde{g}( y_k, \Delta \tau, \alpha)$, 
denote the use of a truncated Fourier
series with truncation parameter $\alpha$ for a fixed value of $N$ so  
that the Fourier series (\ref{g_tilde}) truncates to
\begin{eqnarray}
  \widetilde{g}( y_j, \Delta \tau, \alpha) =
                   \frac{1}{P} \displaystyle \sum_{k=-\alpha N/2}^{\alpha N/2 -1}
                                                     e^{ 2 \pi i \omega_k  j \Delta x} 
                                          \biggl( \frac{ \sin^2 \pi \omega_k \Delta x}
                                                         { ( \pi \omega_k \Delta x)^2}
                                           \biggr)  G(\omega_k, \Delta \tau)~.
            \label{g_truncated}
\end{eqnarray}
Using the notation $\widetilde{g}_{j} (\Delta \tau, \alpha) = \widetilde{g}( y_j, \Delta \tau, \alpha) $, then 
\begin{eqnarray}
     \widetilde{g}_{j+N} (\Delta \tau, \alpha)  = \widetilde{g}_{j} (\Delta \tau, \alpha) \nonumber
\end{eqnarray}
so that our sequence $\{\widetilde{g}_{-N/2} (\Delta \tau, \alpha), \ldots, \widetilde{g}_{N/2-1} (\Delta \tau, \alpha)\}$ is periodic.

{\color{black}
\begin{remark}[Efficient computation of the projections]
It remains to compute the projections. For this one needs to determine the discrete convolution (\ref{g_truncated}). Let
\begin{eqnarray}
   Y_k & = & \biggl( \frac{ \sin^2 \pi \omega_k \Delta x}
                                                         { ( \pi \omega_k \Delta x)^2}
                                           \biggr) G(\omega_k, \Delta \tau)  ~;~
                  k = -\frac{\alpha N}{2}, \ldots, \frac{\alpha N}{2} -1. \nonumber
\end{eqnarray}
Then rewriting  
$e^{ 2 \pi i \omega_k  j \Delta x}  
                                       =  e^{2 \pi i k \ell /( \alpha N)}$  with $ \ell = j\alpha$
and defining
\begin{eqnarray}
  \mathbb{Y}_{\ell} & = & \frac{1}{P} \sum_{k=-\alpha N/2}^{\alpha N/2 -1} 
                            e^{2 \pi i k \ell /( \alpha N)} Y_k~; ~ \ell = -\frac{\alpha N}{2}, \ldots, \frac{\alpha N}{2} -1
             \label{eff_3}
\end{eqnarray}
gives $\{\mathbb{Y}_{\ell} \}$ as the DFT of the $\{Y_k\}$ (of size $ N^{\prime} = \alpha N$).
Consequently, using equations (\ref{g_truncated}) and (\ref{eff_3})
\begin{eqnarray}
    \widetilde{g}( y_j, \Delta \tau, \alpha) =  \mathbb{Y}_{\ell} ~;~ \ell = j \alpha ~;~
                                                      j = -N/2, \ldots, N/2-1  ~.
      \label{setup_1}
\end{eqnarray}
Thus the projections  $\{\widetilde{g}( y_j, \Delta \tau, \alpha)\} $ are computed 
via a single FFT of size $N^{\prime}$.
\end{remark}}

\bigskip
For $k=-\frac{N}{2}, \ldots, \frac{N}{2} -1$ we define $\widetilde{G}(\omega_k, \Delta \tau, \alpha)$ as the DFT of the $\{\widetilde{g}(y_m, \Delta \tau, \alpha)\}$
\begin{eqnarray}
    \widetilde{G}(\omega_k, \Delta \tau, \alpha) & = & \frac{P}{N}
                     \sum_{m=-N/2}^{N/2-1} e^{- 2 \pi i m k/N } ~ \widetilde{g}(y_m, \Delta \tau, \alpha).
        \label{smooth_G_5}
\end{eqnarray}
Note that
\begin{eqnarray}
   \Delta x = \frac{P}{N} > \frac{P}{N^{\prime}} = \frac{P}{ \alpha N} \nonumber
\end{eqnarray}
that is, the  basis function is integrated  over a grid of size $\Delta x > P/N^{\prime}$, and so is
larger than the grid spacing on the $N^{\prime}$ grid.
As $\alpha \rightarrow \infty$, there is no error in evaluating
these integrals (projections) for a fixed value of $N$.  For any finite
$\alpha$, there is an error due to the use of a truncated Fourier Series.

Again, we emphasize that the truncation for the
Fourier series representation of the projection of the Green's function
in (\ref{g_truncated}) does not have to use the
same number of terms ($\alpha N$) as used in the discrete
convolution ($N$).  Instead we can take a very accurate expansion
of the Green's function projection and then translate this
back to the coarse grid using (\ref{smooth_G_5}).
There is no further loss of information in this last step.
As remarked above, we only use the Fourier representation
of $\widetilde{g}( y_j, \Delta \tau, \alpha)$ to carry
out the discrete convolution, that is a dense matrix vector multiply,
efficiently.  The discrete
convolution in Fourier space is exactly equivalent to the
discrete convolution in physical space, assuming periodic
extensions.

{\color{black}
\begin{remark}[Assumption \ref{green_assump_a} revisited]
The assumption that $g(x, x^{\prime}, \Delta \tau) =  g( x - x^{\prime}, \Delta \tau) $ will permit
fast computation of a dense matrix-vector multiply using
an FFT. As mentioned earlier this assumption  holds for
L\'evy processes, but does not hold, for example, for a Heston stochastic
volatility model.  
However the basic idea of projection of the
Green's function onto linear basis functions can be used even if Assumption \ref{green_assump_a}
does not hold. The price, in this case, is 
a loss of computational efficiency.
As an example, in the case of the Heston stochastic volatility model, one has a closed form for the
characteristic function but here
the Green's function has the form $g = g( \nu, \nu^{\prime}, x - x^{\prime}, \Delta \tau)$
where $\nu$ is the variance and $x = \log S$, where $S$ is the asset price.
In this case, we can use an FFT effectively in the $x$ direction, but not in the $\nu$
direction.
\end{remark}
}

\begin{remark}[Relation to the COS method]
In the COS method, the solution $v(x,\tau)$ is also expanded in a Fourier
series. This gives exponential convergence of the entire algorithm for
smooth $v(x,\tau)$ which in turn requires that we have a highly accurate
Fourier representation of $v(x,\tau)$.  However, suppose $v(x,\tau)$ is
obtained from applying an impulse control using a numerical
optimization method at discrete points on a previous step, using 
linear interpolation (the only interpolation method
which is monotone in general). In that case we will not have
an accurate representation of the Fourier series
of $v(x,\tau)$.   In addition,  it does not seem possible to ensure
monotonicity for the COS method. 
So far, we have only assumed that
the $v(x, \tau)$ can be expanded in terms of piecewise linear basis
functions.  This property can be used to guarantee monotonicity.
However convergence will be slower than the COS method if the solution is smooth.
\end{remark}

\begin{remark}[Piecewise constant basis functions]
The equations and previous discussion in this section also holds if our basis functions are piecewise constant functions, that is, basis functions $\phi_j$ which are nonzero over $[x_j - \Delta x/2, x_j + \Delta x/2]$. 
{\color{black}
In this case computing the integral in (\ref{integral}) gives
\begin{eqnarray}
   \widetilde{g}( y_j, \Delta \tau) & = & 
                                            \frac{1}{P} \displaystyle \sum_{k=-\infty}^{\infty}
                                                     e^{ 2 \pi i \omega_k  y_j} 
                                          \biggl( \frac{ \sin \pi \omega_k \Delta x}
                                                         {  \pi \omega_k \Delta x}
                                           \biggr)  G(\omega_k, \Delta \tau) 
          \label{g_tilde2}
\end{eqnarray}}
with the subsequent equations also requiring slight modifications.
\end{remark}

\subsection{Computing the Monotone Scheme}

In order to ensure our monotone approach is effective,  it remains to compute the discrete convolution (\ref{mono_conv}) efficiently.
For the DFT pair for $v_j( \tau)$ and $V( \omega_p, \tau)$, 
we recall that $x_j = \hat{x}_0 + j \Delta x$ 
and so
\begin{eqnarray}
  v_j (\tau) &=& ~~\sum_{\ell = -N/2}^{ N/2-1} V( \omega_\ell, \tau) e^{2 \pi i \omega_l x_j}
                            ~~~=~~  \sum_{\ell = -N/2}^{N/2-1} \biggl( e^{2 \pi i \omega_\ell \hat{x}_0} 
                                                     \biggr)
                           V( \omega_\ell, \tau) e^{2 \pi i j \ell/N} \nonumber \\
        V( \omega_p, \tau) & = & ~~\frac{1}{N} \sum_{\ell = -N/2}^{N/2-1} e^{- 2 \pi i \omega_p x_\ell}
                                           v_\ell( \tau) 
                            ~~= ~~ \frac{1}{N}  \biggl( e^{ - 2 \pi i \omega_p \hat{x}_0} 
                                                     \biggr)
                                            \sum_{\ell = -N/2}^{ N/2-1} e^{- 2 \pi i p \ell /N}
                              v_\ell( \tau)
         ~.
                  \label{DFT_pair_2}
\end{eqnarray}
%

Suppose we write $\widetilde{g}( x_k - x_j, \Delta \tau)$ as a DFT
\begin{eqnarray}
   \widetilde{g}_{k-j} (\Delta \tau, \alpha) 
     & = & \frac{1}{P} \sum_{p= -N/2}^{N/2-1} \widetilde{G}(\omega_p , \Delta \tau, \alpha) e^{2 \pi i (k-j)p/N}
                    ~,
                  \label{Fourier_G}
\end{eqnarray}
where we use equation (\ref{smooth_G_5}) to determine $\widetilde{G}(\omega_p , \Delta \tau, \alpha)$.
Substituting equations (\ref{Fourier_G}) and (\ref{DFT_pair_2}) into equation (\ref{mono_conv}) we then get
\begin{eqnarray}
 v( x_k, \tau + \Delta \tau) &= & \Delta x \sum_{j=-N/2}^{N/2-1} v_j^n ~
     \widetilde{g}( x_k - x_j , \Delta \tau, \alpha) \nonumber \\
     & = & \frac{1}{N} \sum_{p= -N/2}^{N/2-1}  ~~
              \sum_{\ell = -N/2}^{N/2-1}  \biggl( e^{2 \pi i \omega_\ell x_{0}} \biggr) 
                   \widetilde{G}(\omega_p, \Delta \tau, \alpha) V( \omega_\ell, \tau)
                                   e^{2 \pi i   k p/N} \sum_{j= -N/2}^{N/2-1} 
                     e^{2 \pi i  j (\ell-p)/N}  \nonumber \\
     & = & \sum_{p= -N/2}^{N/2-1} 
                   \biggl( e^{2 \pi i \omega_p x_{0}} \biggr) V(\omega_p, \tau) 
                          \widetilde{G}(\omega_p, \Delta \tau, \alpha) 
              e^{2 \pi i k p/N} 
      \label{dis_conv_2}
\end{eqnarray}
where the last equation follows from the classical orthogonality properties of $N^{th}$ roots of unity.

From equation (\ref{DFT_pair_2}) we have
\begin{eqnarray}
        V( \omega_p, \tau) 
                           & = & \frac{1}{N}  \biggl( e^{ - 2 \pi i \omega_p \hat{x}_0} 
                                                     \biggr)
                                            \sum_{\ell = -N/2}^{N/2-1} e^{- 2 \pi i p \ell /N}
                                   v_\ell( \tau) 
                  =  \biggl( e^{ - 2 \pi i \omega_p \hat{x}_0} 
                                                     \biggr) \widetilde{V}( \omega_p, \tau) 
                                \label{mono_conv_2}
\end{eqnarray}
with
\begin{eqnarray} 
                             \widetilde{V}( \omega_p, \tau) & = &  \frac{1}{N} 
                                           \sum_{\ell = -N/2}^{N/2-1} e^{- 2 \pi i p \ell /N}
                                   v_\ell( \tau)   \nonumber     
\end{eqnarray}
the DFT of $\{v_\ell(\tau)\}$. 
Finally substituting equation (\ref{mono_conv_2}) into (\ref{dis_conv_2}) gives
\begin{eqnarray}
   v( x_k, \tau + \Delta \tau)  =   \sum_{p=-N/2}^{N/2-1} \widetilde{V}( \omega_p, \tau) ~
                                                             \widetilde{G}( \omega_p, \Delta \tau, \alpha)
                                            ~e^{ 2 \pi i p \ell /N}
   ~,
         \label{final_mono_conv} 
\end{eqnarray}
which we recognize as the inverse DFT of $\{\widetilde{V}( \omega_p, \tau) ~
                                                             \widetilde{G}( \omega_p, \Delta \tau, \alpha) \}$.
\begin{remark}[Monotonicity]
Equations (\ref{final_mono_conv}) and (\ref{mono_conv}) are algebraic identities (assuming periodic extensions). Hence if we use  (\ref{final_mono_conv}) to advance the solution, then 
this is algebraically identical to using  (\ref{mono_conv}) to advance the solution.  Thus we can analyze the properties of equation (\ref{final_mono_conv})
by analyzing equation (\ref{mono_conv}).  In particular, if  $~~\widetilde{g}(x_k, \Delta \tau, \alpha)~\geq~0$ then the scheme is monotone.
\end{remark}

\begin{remark}[Converting FST or CONV to monotone form]
Equation (\ref{final_mono_conv}) is formally identical
with equation (\ref{convolution_fourier}).  This has the practical
result that any FST or CONV software can be converted to
monotone form by a preprocessing step which computes 
$\widetilde{G}( \omega_p, \Delta \tau, \alpha)$,
and choosing a trapezoidal rule for the integral in equation (\ref{conv_3}).
\end{remark}


\section{Monotone algorithm for solution of the control problem}\label{sec5}

In this section we describe our monotone algorithm for the control problem (\ref{PIDE_1} - \ref{intervention_1}).
Let $(v^{n})^+$ be the vector of values of our solution just after ${\tau_n}$ as defined earlier in equation (\ref{v_vec_def}) 
and $\mathbb{I}_{\Delta x}(x)$ the linear interpolation operator defined as in equation (\ref{interp_def}). Let
\begin{eqnarray}
\widetilde{V}^n & = & [ \widetilde{V}(\omega_{-N/2}, \tau_n), 
               \ldots, \widetilde{V}(\omega_{N/2-1}, \tau_n) ]
     ~ = ~ DFT[ ~ (v^{n})^- ~ ] ~ \nonumber
\end{eqnarray}
and
\begin{eqnarray}
   \widetilde{G} = [ \widetilde{G}( \omega_{-N/2}, \Delta \tau, \alpha), \ldots,
                        \widetilde{G}( \omega_{N/2-1}, \Delta \tau, \alpha) ]~. \nonumber
\end{eqnarray}
  
Let us assume that our Green's function is not an explicit function of $\tau$ but rather we have
$g~=~g( x-x^{\prime}, \Delta \tau)$ and that the time steps are all constant, that is, $\tau_{n+1} - \tau_n = \Delta \tau = const$.
%
In this case we can  compute $\widetilde{G}( \omega_k, \Delta \tau, \alpha)$ only once.  If these two assumptions do 
not hold, then $\widetilde{G}( \cdot )$ would have to be recomputed frequently and hence our algorithm for ensuring monotonicity becomes more costly. 

Algorithm \ref{alg_start} describes the 
computation of $\widetilde{G}( \cdot )$. Here we test for monotonicity (up to a small tolerance) by minimizing the effect of any negative weights which are determined via
$$
\sum_j \Delta x | \min( \widetilde{g}( y_j, \Delta \tau, \alpha) , 0)|  < \epsilon_1 \frac{\Delta \tau }{T} .
$$
The test for accuracy of the projection occurs by the comparison
$$
\max_j \Delta x | \widetilde{g}( y_j, \Delta \tau, \alpha) - 
                              \widetilde{g}( y_j, \Delta \tau, \alpha/2)| < \epsilon_2 .
$$
Both monotonicity and convergence tests 
are scaled by $\Delta x$ so that these quantities are bounded as 
$\Delta \tau \rightarrow 0, \forall \Delta x$ 
(the Green's function becomes unbounded as 
$\Delta \tau \rightarrow 0$, but the integral of 
the Green's function is bounded by unity).  In addition, 
the monotonicity test scales $\epsilon_1$ by $\Delta \tau / T$ in 
order to eliminate the number of timesteps from our monotonicity bounds. This is also discussed further in Section \ref{mono_proof}.

\begin{algorithm}[!h]
\caption{
\label{alg_start}
Initialization of the monotone Fourier method.
}
 \begin{algorithmic}[1]
    \REQUIRE Closed form expression for $G( \omega, \Delta \tau)$, the Fourier transform of the Green's Function
    \STATE Input: $N, \Delta x, \Delta \tau$
    \STATE Let $\alpha = 1$ and compute $\widetilde{g}( y_j, \Delta \tau, 1)$ .
    \FOR[Construct accurate $\widetilde{g}$]{$\alpha = 2^k; k=1,2, \ldots$ until convergence}

    \STATE Compute $\widetilde{g}( y_j, \Delta \tau, \alpha),~
                          \widetilde{G}( \omega_j, \Delta \tau, \alpha) , 
                      ~ j=-\frac{N}{2}, \ldots, \frac{N}{2}-1$ using (\ref{setup_1})-(\ref{smooth_G_5})
    \STATE $test_1 = \displaystyle \sum_j \Delta x \min( \widetilde{g}( y_j, \Delta \tau, \alpha) , 0)$$~~~~~~~~~~~~~~~~~~~~~~~~~~~~~~~~~~~~~~~~~~~~~~~\{ \mbox{ Monotonicity test }\}$\label{test_1}
    \STATE $test_2 = \displaystyle \max_j \Delta x | \widetilde{g}( y_j, \Delta \tau, \alpha) - 
                              \widetilde{g}( y_j, \Delta \tau, \alpha/2)|$ 
                    $~~~~~~~~~~~~~~~~~~~~~~~~~~~~~~~~~~~~~~~~~\{ ~\mbox{ Accuracy test }\}$
              \label{test_2}
    \IF { ($|test_1| <  \epsilon_1 (\Delta \tau /T) ) ~ {\mbox{and}} ~  ( test_2  < \epsilon_2)$ }  
           \label{if_algo1}
            \STATE break from for loop $~~~~~~~~~~~~~~~~~~~~~~~~~~~~~~~~~~~~~~~~~~~~~~~~~~~~~~~~~~~~~~~~\{ \mbox{ Convergence test }\}$
    \ENDIF
   \ENDFOR \COMMENT{End accurate $\widetilde{g}$ loop}
   \STATE Output: Weights $\widetilde{G}( \omega_j, \Delta \tau, \alpha), ~ j=-\frac{N}{2}, \ldots, \frac{N}{2}-1$ in Fourier domain.
 \end{algorithmic}
\end{algorithm}
In Algorithm \ref{alg_start}, the test on line \ref{test_1} will ensure that monotonicity
holds to a user specified tolerance and the test on line \ref{test_2} 
ensures accuracy of the projections. 
The complete monotone algorithm for the control problem is given in Algorithm \ref{alg_monotone}. 

\begin{remark}[Convergence of Algorithm \ref{alg_start}.]
In Appendix \ref{projection_convergence} we show that for typical
Green's functions,  the test for monotonicity on line \ref{test_1} in
Algorithm \ref{alg_start}  and the accuracy test on line \ref{test_2}
are usually satisfied 
for $\alpha =2, 4$, for typical values of $\epsilon_1, \epsilon_2$.  
\end{remark}

\begin{remark}[Complexity]
The complexity of using (\ref{final_mono_conv}) to advance the time (excluding the cost of determining an optimal control) is $O(N \log N)$ operations, roughly the same as the usual FST/CONV methods.
\end{remark}

\begin{algorithm}[!h]
\caption{
\label{alg_monotone}
Monotone Fourier method.
}
 \begin{algorithmic}[1]
    \REQUIRE Weights $\widetilde{G} = \{\widetilde{G}( \omega_j, \Delta \tau, \alpha) \}, \mbox{ for } ~ j=-\frac{N}{2}, \ldots, \frac{N}{2}-1$ in Fourier domain 
               (from Algor.~\ref{alg_start}).
    \STATE Input: number of timesteps $M$, initial solution $(v^0)^-$
    \STATE $(v^0)^+ = \inf_c \mathcal{M}(c) \bigl( ~\mathbb{I}_{\Delta x} (x) (v^0)^- ~ \bigr)$
    \FOR[Timestep loop]{$m=1,\ldots, M$}

    \STATE $\widetilde{V}^{m-1} = FFT[~ (v^{m-1})^+ ~]~~~~~~~~~~~~~~~~~~~~~~~~~~~~~~~~~~~~~~~~~~~~~~~~~~~~~~~~~~~~~$\COMMENT{Frequency  domain}
             \label{DFT_v}
    \STATE $(v^{m})^- = iFFT[~ \widetilde{V}^{m-1} \circ \widetilde{G} ~]~~~~~~~~~~~~~~~~~~~~~~~~~~~~~~~~~~~~~~~~~~~~~~~~~~~~~~~~~~$ \COMMENT{Physical  domain}
                \label{IDFT_v}
    \STATE $ v( x_j, \tau_{m}^+) = \inf_c \mathcal{M}(c) 
              \bigl( ~\mathbb{I}_{\Delta x} (x_j) (v^{m})^- ~\bigr) ~;~ j=-\frac{N}{2}, \ldots, \frac{N}{2}-1 ~~~~~~~~~~~~~~~~$ \COMMENT{Optimal control}
                             \label{optimal_control_algo}

   \ENDFOR \COMMENT{End timestep loop}
 \end{algorithmic}
\end{algorithm}

\section{Properties of the Monotone Fourier Method}\label{mono_proof}

In this section we prove a number of properties satisfied by 
our $\epsilon$-monotone Fourer algorithm. The main properties 
include  $\ell_{\infty}$ stability and a type of $\epsilon$-discrete comparison principle.

\begin{lemma}\label{sum_g}
Let $C_1$ be a constant such that the exact Green's function satisfies $C_1 = \int_{\mathbb{R}} g(x, \Delta \tau)~dx $. Then for all $k$
\begin{eqnarray}
   \Delta x \sum_{j=-N/2}^{N/2-1} 
        \widetilde{g}( x_k - x_j, \Delta x, \alpha)
     & = & C_1~~ \mbox{ with } ~~\Delta x = \frac{P}{N} ~. \nonumber 
\end{eqnarray}
\end{lemma}
\begin{proof}
\begin{eqnarray}
   \Delta x \sum_{j=-N/2}^{N/2-1} \widetilde{g}( x_k - x_j, \Delta x, \alpha)
    & = & \frac{P}{N} \sum_{\ell=-N/2}^{N/2-1} 
                 \widetilde{g}(y_\ell, \Delta \tau, \alpha) ~;~\mbox{ where }  y_\ell = x_k - x_j~, ~\ell = k - j  \nonumber \\
    & = & \frac{P}{N} \sum_{\ell=-N/2}^{N/2-1} \frac{1}{P}
                  \sum_{k=-\alpha N/2}^{\alpha N/2-1}
                        e^{ 2 \pi i \omega_k   \ell \Delta x}
                                          \biggl( \frac{ \sin^2 \pi \omega_k \Delta x}
                                                         { ( \pi \omega_k \Delta x)^2}
                                           \biggr) G(\omega_k, \Delta \tau) \nonumber \\
      & = & \frac{1}{N} \sum_{k=-\alpha N/2}^{\alpha N/2-1} 
                       \biggl( \frac{ \sin^2 \pi \omega_k \Delta x}
                                                         { ( \pi \omega_k \Delta x)^2}
                                           \biggr) G(\omega_k, \Delta \tau) 
                  \sum_{\ell=-N/2}^{N/2-1}  e^{2 \pi i \ell k/N} ~ \nonumber \\
     & = & ~ G( 0 , \Delta \tau) ~=~ \int_{-\infty}^{\infty} g(x, \Delta \tau) ~dx  =  C_1 ~. \nonumber
\end{eqnarray}
\end{proof}

\begin{theorem}[$\ell_{\infty}$ stability]\label{stab_lemma}
Assume that  $\widetilde{G}$ is computed using Algorithm \ref{alg_start}, 
that $(v^{n})^- $ is computed
from
\begin{eqnarray}
   (v_k^{n})^- & = &  \sum_{j=-N/2}^{N/2-1}  \Delta x ~ \widetilde{g}_{k-j} (v_j^{n-1})^+
  ~,  \label{eq_inf_1}
\end{eqnarray}
and that
\begin{eqnarray}
   \| (v^{n})^+ \|_{\infty} & \leq & \| (v^{n})^- \|_{\infty} ~.  \label{stab_1}
\end{eqnarray} 
Then for every $0 \leq  n \leq M$ we have
\begin{eqnarray}
  \| (v^{n})^+ \|_{\infty} \leq  C_2 = e^{2  \epsilon_1 } \| (v^0)^- \|_{\infty}   ~. \nonumber
\end{eqnarray}
\end{theorem}
\begin{proof}
From equation (\ref{eq_inf_1}) 
\begin{eqnarray}
  (v_k^{n})^- & = &  \sum_{j=-N/2}^{N/2-1}  \Delta x ~ \widetilde{g}_{k-j} (v_j^{n-1})^+ \nonumber \\
                & = &  \sum_{j=-N/2}^{N/2-1}  \Delta x ~\max( \widetilde{g}_{k-j} ,0) (v_j^{n-1})^+ 
                     +  \sum_{j=-N/2}^{N/2-1}  \Delta x~ \min( \widetilde{g}_{k-j} ,0) (v_j^{n-1})^+
              ~.  \label{stab_2}
\end{eqnarray}
which then implies
\begin{eqnarray}
         | (v_k^{n})^- |& \leq & \| (v^{n-1})^+ \|_{\infty} 
                                \sum_{j=-N/2}^{N/2-1}  \Delta x ~ \max( \widetilde{g}_{k-j} , 0) 
       +   \| (v^{n-1})^+ \|_{\infty}  \sum_{j=-N/2}^{N/2-1}   \Delta x ~  | \min( \widetilde{g}_{k-j} ,0) |  
                ~. \nonumber 
\end{eqnarray}
From Lemma \ref{sum_g}
\begin{eqnarray}
    \sum_{j=-N/2}^{N/2-1}  \Delta x ~ \max( \widetilde{g}_{k-j} , 0) 
              & = & C_1 +  \sum_{j=-N/2}^{N/2-1}  \Delta x ~ | \min( \widetilde{g}_{k-j} , 0) |
         \label{stab_4} ~
\end{eqnarray}
and so
\begin{eqnarray}
    | (v_k^{n})^- |& \leq & \| (v^{n-1})^+ \|_{\infty} \biggl( 
                                   C_1 +  2 \sum_{j=-N/2}^{N/2-1} \Delta x ~  | \min( \widetilde{g}_{k-j} , 0) |
                                              \biggr) \nonumber \\
                                              & \leq & \| (v^{n-1})^+ \|_{\infty} ( C_1 + 2  \epsilon_1 \frac{\Delta \tau}{T} )
                                              \label{stab_6}~ .
\end{eqnarray}
Using lines \ref{test_1} and \ref{if_algo1} in Algorithm \ref{alg_start}. 
Since equation (\ref{stab_6}) is true for any $k$ we have that
\begin{eqnarray}
   \| (v^{n})^- \|_{\infty} & \leq & \| (v^{n-1})^+ \|_{\infty} ( C_1 + 2 \epsilon_1 \frac{\Delta \tau}{T}) ~, \nonumber
\end{eqnarray}
which combined with equation (\ref{stab_1})  and using $C_1 \leq 1$ gives
\begin{eqnarray}
 \| (v^{n})^+ \|_{\infty} &\leq & 
\| (v^{n-1})^+ \|_{\infty} ( 1 + 2 \epsilon_1 \frac{\Delta \tau}{T}) ~.
         \nonumber 
\end{eqnarray}
Iterating the above bound and using equation (\ref{stab_1}) at $n=0$ gives
\begin{eqnarray}
 \| (v^{n})^+ \|_{\infty} & \leq & \| (v^0)^- \|_{\infty} (1 + 2  \epsilon_1 \frac{\Delta \tau}{T} )^{n} 
                                \nonumber \\
                   &  \leq &  \| (v^0)^- \|_{\infty} ~e^{ 2 \epsilon_1  n \frac{\Delta \tau}{T} }
                   ~~ \leq ~~  \| (v^0)^- \|_{\infty} ~e^{ 2  \epsilon_1 } ~=~C_2~.
                            \nonumber 
\end{eqnarray}

\end{proof}

\begin{remark}[Jump condition]\label{jump_remark}
We remark that (\ref{stab_1}), the jump condition $\| (v^{n})^+ \|_{\infty}  \leq  \| (v^{n})^- \|_{\infty} $, is trivially satisfied if 
$\mathbb{I}_{\Delta x} (x)$  in line \ref{optimal_control_algo} in Algorithm \ref{alg_monotone}
is a linear interpolant.
\end{remark}
\noindent
From Theorem \ref{stab_lemma} and Remark \ref{jump_remark}
we have the immediate result,
\begin{corollary}[Stability of Algorithm \ref{alg_monotone}]
Algorithm \ref{alg_monotone} is $\ell_{\infty}$ stable.
\end{corollary}

\begin{lemma}[Minimum value of solution.]\label{min_lemma}
Let $ (v^n)^+$ be generated using equation  \ref{eq_inf_1}, and set
\begin{eqnarray}
   (v^n)^+_{\min} = \min_k ~(v_k^n)^+ ~. \nonumber 
\end{eqnarray}
If the conditions for Lemma \ref{stab_lemma} are satisfied and
\begin{eqnarray}
   (v^n)^+_{\min} & \geq & (v^n)^-_{\min} ~,\label{min_eq_1}
\end{eqnarray}
then
\begin{eqnarray}
  (v^n)^+_{\min} & \geq & (v^0)^-_{\min} ~ (C_3)^n - C_2 (e^{\epsilon_1 } -1) 
                     \nonumber 
\end{eqnarray}
where $C_2 ~=~ \| (v^0)^- \|_{\infty} e^{ 2  \epsilon_1 }$ is given in  Lemma \ref{stab_lemma} and $C_3 =   \sum_{j=-N/2}^{N/2-1}  \Delta x ~\max( \widetilde{g}_{k-j} ,0)$.
\end{lemma}
\begin{proof}
From equation (\ref{eq_inf_1}) and using equation (\ref{stab_2}) along with the definition of $C_3$
we obtain
\begin{eqnarray}
   (v_k^n)^- & \geq & (v^{n-1})^+_{\min}  \sum_{j=-N/2}^{N/2-1}  \Delta x ~\max( \widetilde{g}_{k-j} ,0)
             + \sum_{j=-N/2}^{N/2-1}  \Delta x ~\min( \widetilde{g}_{k-j} ,0) (v_j^n)^+ 
                          \nonumber \\
              & \geq & (v^{n-1})^+_{\min}  \sum_{j=-N/2}^{N/2-1}  \Delta x ~\max( \widetilde{g}_{k-j} ,0)
                    - \| (v^{n-1})^+ \|_{\infty}  \sum_{j=-N/2}^{N/2-1}  \Delta x ~ | \min( \widetilde{g}_{k-j} ,0) |~
                     \nonumber \\
                    & =  &  (v^{n-1})^+_{\min} C_3
                    - \| (v^{n-1})^+ \|_{\infty} \biggl(
                                       \sum_{j=-N/2}^{N/2-1}  \Delta x ~| \min( \widetilde{g}_{k-j} , 0) |
                                                 \biggr) ~.
                    \nonumber 
\end{eqnarray}
Using Lemma \ref{stab_lemma} and lines  \ref{test_1} and \ref{if_algo1} in Algorithm \ref{alg_start} then gives
\begin{eqnarray}
    (v_k^n)^- & \geq &  (v^{n-1})^+_{\min} C_3
                            - C_2 \epsilon_1  \frac{\Delta \tau}{T}  ~
                      \nonumber 
\end{eqnarray}
and, since this is valid for any $k$,  using (\ref{min_eq_1}) we obtain
 \begin{eqnarray}
    (v^n)_{\min}^+ & \geq & (v^{n-1})^+_{\min} C_3
                            - C_2 \epsilon_1 \frac{\Delta \tau}{T} . \nonumber
                            \end{eqnarray}
Iterating implies
\begin{eqnarray}
     (v^n)_{\min}^+ & \geq & (v^{0})^+_{\min} {C_3}^n - C_2 \epsilon_1  \frac{\Delta \tau }{T} 
                    \biggl( \frac{ 1 - {C_3}^n}{ 1 - C_3} \biggr) \nonumber \\
                  & \geq & (v^{0})^-_{\min}  {C_3}^n - C_2 \epsilon_1  \frac{\Delta \tau}{T} 
                    \biggl( \frac{ 1 - {C_3}^n}{ 1 - C_3} \biggr)   
            ~, \label{min_eq_6}
\end{eqnarray}
where we again use equation (\ref{min_eq_1}) in the last line.
From equation (\ref{stab_4}) and the definition of $C_3$ we have 
\begin{eqnarray}
     C_3 &=& C_1 +  \sum_{j=-N/2}^{N/2-1}  \Delta x ~| \min( \widetilde{g}_{k-j} , 0) |
                 \leq  1 + \epsilon_1  \frac{\Delta \tau}{T} ~, 
              \label{min_eq_7}
\end{eqnarray}
where the last inequality  follows  
lines \ref{test_1} and \ref{if_algo1} in Algorithm \ref{alg_start} (and recalling that $C_1 \leq 1$).
Combining equations (\ref{min_eq_1}), (\ref{min_eq_6}) and (\ref{min_eq_7} ) and noting that $n \Delta \tau \leq T$ gives
\begin{eqnarray}
    (v^n)_{\min}^+ & \geq & (v^{0})^+_{\min} {C_3}^n 
                          - C_2 ( e^{\epsilon_1 } - 1)~ 
                           ~~\geq ~~ (v^{0})^-_{\min} {C_3}^n 
                          - C_2 ( e^{\epsilon_1 } - 1)~.\nonumber
\end{eqnarray}
\end{proof}

\begin{remark}\label{interp_remark}
We note that condition \ref{min_eq_1}, that is, $(v^n)^+_{\min}  \geq  (v^n)^-_{\min}$ is satisfied if
$\mathbb{I}_{\Delta x} (x)$  in line \ref{optimal_control_algo} in Algorithm \ref{alg_monotone}
is a linear interpolant.
\end{remark}
\bigskip

\begin{theorem}[$\epsilon$-Discrete Comparison Principle] \label{main_result}
Suppose we have two independent discrete solutions 
\begin{eqnarray}
    (u^{n})^+ &= &[ u(x_{-N/2}, \tau_n^+), \ldots, u(x_{+N/2-1}, \tau_n^+) ]
                   \nonumber \\
    (w^{n})^+ &= &[  w(x_{-N/2}, \tau_n^+), \ldots, w(x_{+N/2-1}, \tau_n^+) ]
\end{eqnarray}
with
\begin{eqnarray}
     (u^{0})^- & \geq & (w^{0})^-
\end{eqnarray}
where the inequality is understood in the component-wise sense,
and  $(u^{n})^+, (w^{n})^+$ are computed using Algorithm \ref{alg_monotone}.
If $\widetilde{G}$ is
computed using Algorithm \ref{alg_start} and 
$\mathbb{I}_{\Delta x} (x)$ is a linear interpolant 
then
\begin{eqnarray}
(u^{n})^+ - (w^{n})^+ & \geq 
                        -   \epsilon_1 \| (u^0 - w^0)^- \|_{\infty}  + O( \epsilon_1^2) ~;~ \epsilon_1 \rightarrow 0~.
\end{eqnarray}
\end{theorem}
\begin{proof}
Let $ (z^n)^+ = (u^{n})^+ - (w^{n})^+$, $ (z^n)^- = (u^{n})^- - (w^{n})^- $,
then  
\begin{eqnarray*}
  (z_k^{n})^- & = &  \sum_{j=-N/2}^{N/2-1}  \Delta x ~ \widetilde{g}_{k-j} (z_j^{n-1})^+
   ~. 
\end{eqnarray*}
Noting that
\begin{eqnarray}
   z_j(\tau_n^+) = \displaystyle \inf_c \mathcal{M} (c) \bigl( ~\mathbb{I}_{\Delta x} (x_j) (u^{m})^- ~\bigr)
                  - \displaystyle \inf_c \mathcal{M} (c) \bigl( ~\mathbb{I}_{\Delta x} (x_j) (w^{m})^- ~\bigr)
\end{eqnarray}
then
\begin{eqnarray}
   | z_j(\tau_n^+) | &\leq & \displaystyle \sup_c  \mathcal{M}(c)  \big|  ~\mathbb{I}_{\Delta x}~ (x_j) \bigl( (u^{m})^- 
                                                   - (w^{m})^- ~\bigr) ~\big|
\end{eqnarray}
hence, using the definition of the intervention operator (\ref{intervention_def}), 
we obtain
\begin{eqnarray}
   \| (z^n)^+ \|_{\infty} \leq \| (z^n)^- \|_{\infty} ~. 
\end{eqnarray}
Similarly
\begin{eqnarray}
 (z^n)^+_{\min} & = &   \displaystyle \min_j z_j(\tau_n^+) \nonumber \\
               & \geq &  \displaystyle \min_j \displaystyle \inf_c \mathcal{M} (c)  
                                ~\mathbb{I}_{\Delta x} (x_j)  \bigl((u^{m})^-
                                                         - (w^{m})^- ~\bigr) \nonumber \\
               & \geq & (z^n)^-_{\min} ~.
\end{eqnarray}
Hence condition (\ref{stab_1}) of Lemma \ref{stab_lemma} and condition (\ref{min_eq_1}) of Lemma \ref{min_lemma}
are satisfied.  Applying Lemma \ref{min_lemma} to $(z^n)^+, (z^n)^-$ we get
\begin{eqnarray}
    (z^n)^+_{\min} & \geq & (z^0)^-_{\min} (C_3)^n -  e^{2 \epsilon_1} \| (u^0 - w^0)^- \|_{\infty} ( e^{\epsilon_1} -1 )
\end{eqnarray}
where $C_3 = \sum_{j=-N/2}^{N/2-1}  \Delta x ~\max( \widetilde{g}_{k-j} ,0)$.
Since $(z^0)^-_{\min} \geq 0$ and  $0 \leq C_3^n \leq  e^{\epsilon_1}$, the result follows.
\end{proof}

\begin{remark}
If Algorithm \ref{alg_start} is used to construct $\widetilde{G}$ for use
in Algorithm \ref{alg_monotone}, then the $\epsilon$-discrete comparison
property 
is satisfied for any $N, \Delta \tau, M$ up to order $\epsilon_1$.   
Since  typically $\widetilde{g}( y_j, \Delta \tau, \alpha) \rightarrow \widetilde{g}( y_j, \Delta \tau, \infty) \geq 0$
exponentially in $\alpha$, in practice it is very inexpensive to make $\epsilon_1$ as
small as desired.
\end{remark}


\begin{remark}[Continuously observed impulse control problems]
By determining the optimal control at each timestep, we can apply our monotone Fourier
method to the continuously observed impulse control problem
\begin{eqnarray}
  \max \biggl[ v_\tau - \mathcal{L} v, v - \inf_c \mathcal{M}(c) v \biggr] = 0 ~.
\end{eqnarray}
This is effectively a method whereby the optimal control is applied explicitly, as in \citep{chen-forsyth:2008}.  
Using the methods developed in this paper combined with those from  \citep{chen-forsyth:2008}, it is 
straightforward to show that the $\epsilon$-monotone Fourier technique is 
$\ell_{\infty}$ stable and consistent in the viscosity sense as 
$\Delta \tau, \Delta x \rightarrow 0$.  The $\epsilon$-monotone Fourier method is also monotone to $O(h)$ where  $h = O( \Delta x) = O(\Delta \tau)$ is the discretization parameter.  Thus it is possible to show convergence to the
viscosity solution using the results in \citet{barles-souganidis:1991} extended as in \citet{Azimzadeh},
using the $\epsilon$ monotonicity property as in \citet{Bokanowski}.
\end{remark}

\section{Minimization of wrap-around error}\label{wrap_section}


The use of the convolution form for our solution (\ref{final_mono_conv}) is rigorously correct for
a periodic extension of the solution and the Green's function.
In normal option pricing applications, the {\em wrap-around} error
due to periodic extension causes little error. However, in
control applications, the values used in the optimization
step (\ref{intervention_1}) may be near the ends of the grid and
hence large errors may result \citep{Lippa2013,Ruijter2013,Song2016}.
Hence we need to consider methods to reduce errors associated with wrap-around.

In order to minimize the effect of wrap-around we proceed in the following manner.
Given the localized problem on $[x_{\min}, x_{\max}]$ with $N$
nodes, we construct an auxiliary grid with  $N^a = 2N$ nodes,
on the domain $[x_{\min}^a, x_{\max}^a]$ where
\begin{eqnarray}
    x_{\min}^a &= &x_{\min} - \frac{ (x_{\max} - x_{\min}) }{ 2} ~~\mbox{ and } ~~
    x_{\max}^a  =  x_{\max} + \frac{ (x_{\max} - x_{\min} )} { 2}
\end{eqnarray}
with $( x_{\max}^a - x_{\min}^a ) = 2 ( x_{\max} - x_{\min})$. 
We construct and store the DFT of the projection of the Green's function 
$\widetilde{G}( \omega_p, \Delta \tau, \alpha), p = -N^a/2, \ldots, p = N^a/2 -1$
on this auxiliary grid.  We then replace line \ref{DFT_v} in Algorithm \ref{alg_monotone}
by applying the DFT to the solution $v$ on the auxiliary grid
\begin{eqnarray}
     v( x_k, \tau_n^+)^a & = &  v( x_k, \tau_n^+) ~;~ ~~~~~k=-N/2, \ldots, k=N/2-1 \nonumber \\
                         & = & v( x_{-N/2}, \tau_n^+) ~;~~ k=-N^a/2,\ldots, -N/2-1 
                                        \label{aux_small}
                                            \\
                         & = & A( x_{k}, \tau_n^+) ~;~~~~~~ k = N/2, \ldots, N^a/2 -1
                                     \label{aux_large}
  ~,
\end{eqnarray}
where $A(x, \tau)$ is an asymptotic form of the solution, which we assume to be available
from financial reasoning.  On the auxiliary grid near $x \rightarrow -\infty$ we simply 
extend the solution by the constant value at $x = x_{\min}$, which is expected to generate 
a small error, since the grid spacing (in terms of $S=e^x$) is very small.  
We then carry out lines \ref{DFT_v} - \ref{IDFT_v} of Algorithm \ref{alg_monotone}
on the auxiliary grid and generate $ (v^n)^-$ by discarding all the values on the auxiliary 
grid which are not on the  original grid (as these are contaminated by wrap-around errors).
The errors incurred by using extensions (\ref{aux_small}) and (\ref{aux_large}) can be
made small by choosing $ |x_{\min}|$ and $ x_{\max}$ sufficiently large.
\begin{remark}[Use of asymptotic form to reduce wrap-around error]
Use of the above technique necessitates some changes to the proof
of Theorem \ref{main_result}.  However, the main result is the same,
with adjustments to some of the constants in the bounds. This
is a tedious algebraic exercise which we omit.
\end{remark}
{\color{black}
\begin{remark}[Additional complexity to reduce wrap-around]
For a one dimensional problem, the complexity for one timestep is $O(N^a \log N^a) = O( 2N \log( 2N)$,
where $N$ is the number of nodes in the original grid.  In the case of the
path dependent problem in Section \ref{sec8}, if there are $N_x$ nodes in the $\log S$ direction
and $N_b$ nodes in the bond direction, then the complexity for one timestep
is $O( 2 N_b N_x \log( 2N_x) )$.
\end{remark}}

\section{Numerical examples}\label{sec7}

\subsection{European option}

Consider a European option written on an underlying stock whose price $S$ follows a jump diffusion process.  
Denote by $\xi$ the random number representing the jump multiplier so that when a jump occurs, 
we have $S_t =   \xi S_{t^-}$.
The risk neutral process followed by $S_t$ is
\begin{eqnarray}
   \frac{dS_t}{S_{t^-}} &= (r -\lambda \kappa) dt + \sigma  dZ + d
     \biggl( \displaystyle \sum_{i=1}^{\pi_t} (\xi_i-1) \biggr) ~~\mbox{ with } ~~\kappa = E[\xi] - 1
    \label{jump_process}
\end{eqnarray}
where $E[\cdot]$ denotes the expectation operator.  Here, $dZ$ is the increment of a Wiener process,
$r$ is the risk free rate, $\sigma$ is the volatility,
$\pi_t$ is a Poisson process with positive intensity parameter
$\lambda$, and $\xi_i$ are i.i.d.
positive random variables.  The density function $f(y)$,  $y = \log( \xi)$ is
assumed double exponential \citep{Kou2004}
\begin{equation}
f(y) = p_{u} ~\eta_1 ~e^{-\eta_1 y} {\bf{1}}_{y \geq 0} ~+~
  (1-p_{u}) ~\eta_2 ~ e^{\eta_2 y} {\bf{1}}_{y < 0} ~
\label{kou_jumps}
\end{equation}
with the expectation
\begin{eqnarray}
   E[\xi] = \frac{p_{u} ~ \eta_1}{\eta_1 - 1} + \frac{(1-p_{u} ) ~ \eta_2}{\eta_2 + 1} ~.
\end{eqnarray}
Given that a jump occurs, $p_{u}$ is the probability of an upward jump
and $(1-p_{u})$ is the probability of a downward jump.

The price of a European call option $v(x , \tau)$ with $x  = \log S$ is then given as the solution to
\begin{eqnarray}
  v_{\tau} & = & \frac{\sigma^2}{2} v_{xx} + (r - \frac{\sigma^2}{2} - \lambda \kappa) v_x
                  - (r+\lambda) v + \lambda \int_{-\infty}^{+\infty} v( x + y) f(y) ~dy 
                      \nonumber \\
          && ~~\mbox{ with } ~~ v(x, 0)  =   \max(e^x - K, 0) ~. \label{pde_1}
\end{eqnarray}
The Green's function for this problem is given in Appendix A.

  
The particular parameters for this test are given in Table \ref{simple_euro} with the results appearing in Table \ref{simple_euro_results_1}.
All methods obtain smooth second order convergence, with the
exception of the FST/CONV  Simpson rule method which gives fourth
order convergence, due to the higher order quadrature method.
This is to be expected in this case since there is a node
at the strike.  Increasing $x_{\max},  |x_{\min}|$ altered results
in the last 2 digits in the table.  
{\color{black}
This is due to the effect of localizing the problem to $[x_{\min}, x_{\max}]$,
and the effects of FFT wrap-around.}

\begin{table}[tb]
  \begin{center}
    \begin{tabular}{ll} \toprule
      Expiry time & .25 years \\
      Strike K  & 100\\
      Payoff    & call \\
      Initial asset price $S_0$  & 100\\
      Risk-free rate $r$ & .05  \\
      Volatility $\sigma$  & .15 \\
       $\lambda$  & .1 \\
      $\eta_1$ &  3.0465 \\
      $\eta_2$ &  3.0775\\
      $p_{u}$ &  0.3445\\
      $x_{\max}$ & $ \log(S_0) + 10$\\
      $x_{\min}$ & $ \log(S_0) - 10$\\
      $\epsilon_1, \epsilon_2$ & $10^{-6}$ \\
      Asymptotic form $x \rightarrow \infty$   & $A(x) = e^x$ \\
           \bottomrule
    \end{tabular}
    \caption{European call option test.
     \label{simple_euro}}
  \end{center}
\end{table}

\begin{table}[tb]
  \begin{center}
    \begin{tabular}{|l|l|l|l|l|l|l|l|l|} \hline
                & \multicolumn{4}{c|}{Monotone Methods} & \multicolumn{4}{c|}{FST/CONV}\\ \hline
                & \multicolumn{2}{c|}{Piecewise linear} &  \multicolumn{2}{c|}{Piecewise constant}
                & \multicolumn{2}{c|}{Trapezoidal} & \multicolumn{2}{c|}{Simpson} \\
                         \hline
    $N$ & Value & Ratio & Value & Ratio & Value & Ratio & Value & Ratio\\
                          \hline
    $2^9 $      & 3.9808516210  &       & 3.9443958729   &      & 3.9075619850 &       & 3.9784907318   &  \\
    $2^{10}$      & 3.9753205007  &     & 3.9662547470  &       & 3.9571661688 &       & 3.9737010716   &   \\
    $2^{11}$      & 3.9739391670 &  4.0 & 3.9716756819  & 4.0   & 3.9694107823 & 4.1   & 3.9734923202   & 23  \\
    $2^{12}$      & 3.9735939225 &  4.0 & 3.9730282349  &  4.0  & 3.9724624589 & 4.0   & 3.9734796846   & 17  \\
    $2^{13}$      & 3.9735076171 &  4.0 & 3.9733662066  &  4.0  & 3.9732247908 & 4.0   & 3.9734789013   & 16     \\
    $2^{14}$      & 3.9734860412 &  4.0 & 3.9734506895  &  4.0  & 3.9734153372 & 4.0   & 3.9734788524   & 16  \\
            \hline
    \end{tabular}
    \caption{European call option test: value at $\tau=0, S=S_0$.  Parameters in
             Table \ref{simple_euro}.  $N=$ number of nodes. Ratio is the
             ratio of successive changes.
     \label{simple_euro_results_1}}
  \end{center}
\end{table}

In order to stress these Fourier methods, we repeat this example, except now using an expiry time of $T = .001$.  
{\color{black}
Since the Green's function in the physical space converges
to a delta function as $T \rightarrow 0$, we can expect that this
will be challenging for Fourier methods as a large number of terms will be required in the Fourier series in order to
get an accurate representation of the Green's function in the physical
space.  }
The results for this test are shown in Table \ref{simple_euro_results_small}.
The monotone method with piecewise linear basis functions gives reasonable
results for all grid sizes.  The standard FST/CONV methods are
quite poor, except for very large numbers of nodes.  Indeed, using Simpson's rule 
on coarse grids even results in values
larger than $S_0 $ at $S = S_0 = 100$, which violates the provable 
bound for a call option.

\begin{table}[tb]
  \begin{center}
    \begin{tabular}{|l|l|l|l|l|l|l|l|l|} \hline
                & \multicolumn{4}{c|}{Monotone Methods} & \multicolumn{4}{c|}{FST/CONV}\\ \hline
                & \multicolumn{2}{c|}{Piecewise linear} &  \multicolumn{2}{c|}{Piecewise constant}
                & \multicolumn{2}{c|}{Trapezoidal} & \multicolumn{2}{c|}{Simpson} \\
                         \hline
    $N$ & Value & Ratio & Value & Ratio & Value & Ratio & Value & Ratio\\
                          \hline
    $2^9 $        & .19662316859 &      & .94284763015  &       &.24774086499  &      & 319.45747026    &  \\
    $2^{10}$      & .19467436458 &      & .041410269769 &       & .21909081933  &      &521.62802838    &   \\
    $2^{11}$      & .19376651687 & 2.1  & .15335986938  & -8.0  & .18611676723  &.87   &439.13444172    & -2.5   \\
    $2^{12}$      & .19346709107 & 3.0  & .18477993505  & 3.6   & .17728640855  &3.7   &27.002978049   &  0.2  \\
    $2^{13}$      & .19339179620 & 4.0  & .19127438852  & 4.8   & .18913280108  &-.75  &.19367805822   &  15     \\
    $2^{14}$      & .19337297842 & 4.0  & .19284673379  & 4.1   & .19231903134  &3.7   &.19338110881  &  $9\times 10^{4}$  \\
            \hline
    \end{tabular}
    \caption{European call option test: value at $\tau=0, S=S_0$.  Parameters in
             Table \ref{simple_euro} but $T=.001$.  $N=$ number of nodes. Ratio is the
             ratio of successive changes.
     \label{simple_euro_results_small}}
  \end{center}
\end{table}

This phenomenon can be explained by examining Figure \ref{Green_figs}, which
shows the projection of the  Green's functions for the monotone
method (piecewise linear basis function) and the truncated Green's function
for the FST/CONV method.  The projection of the  Green's function for the
monotone method in Figure \ref{Mono_green_fig} clearly has the expected
properties: very peaked near $x=0$ and non-negative for all $x$.  In contrast,
the FST/CONV numerical Green's function is oscillatory and negative
for some values of $x$.
{\color{black}
Figure \ref{Wiggle_plt} shows the FST/CONV (trapezoidal) solution
compared to the Monotone (piecewise linear) solution, on a coarse
grid with $512$ nodes.  The monotone
solution can never produce a value less than zero (to within the tolerance).
Note that monotonicity is clearly violated for the FST/CONV solution,
with negative values for a call option.
The oscillations are even more pronounced if Simpson's quadrature
is used for the FST/CONV method.}

\begin{figure}[tb]
\centerline{%
  \begin{subfigure}[b]{.45\linewidth}
    \centering
     \includegraphics[width=3.in]{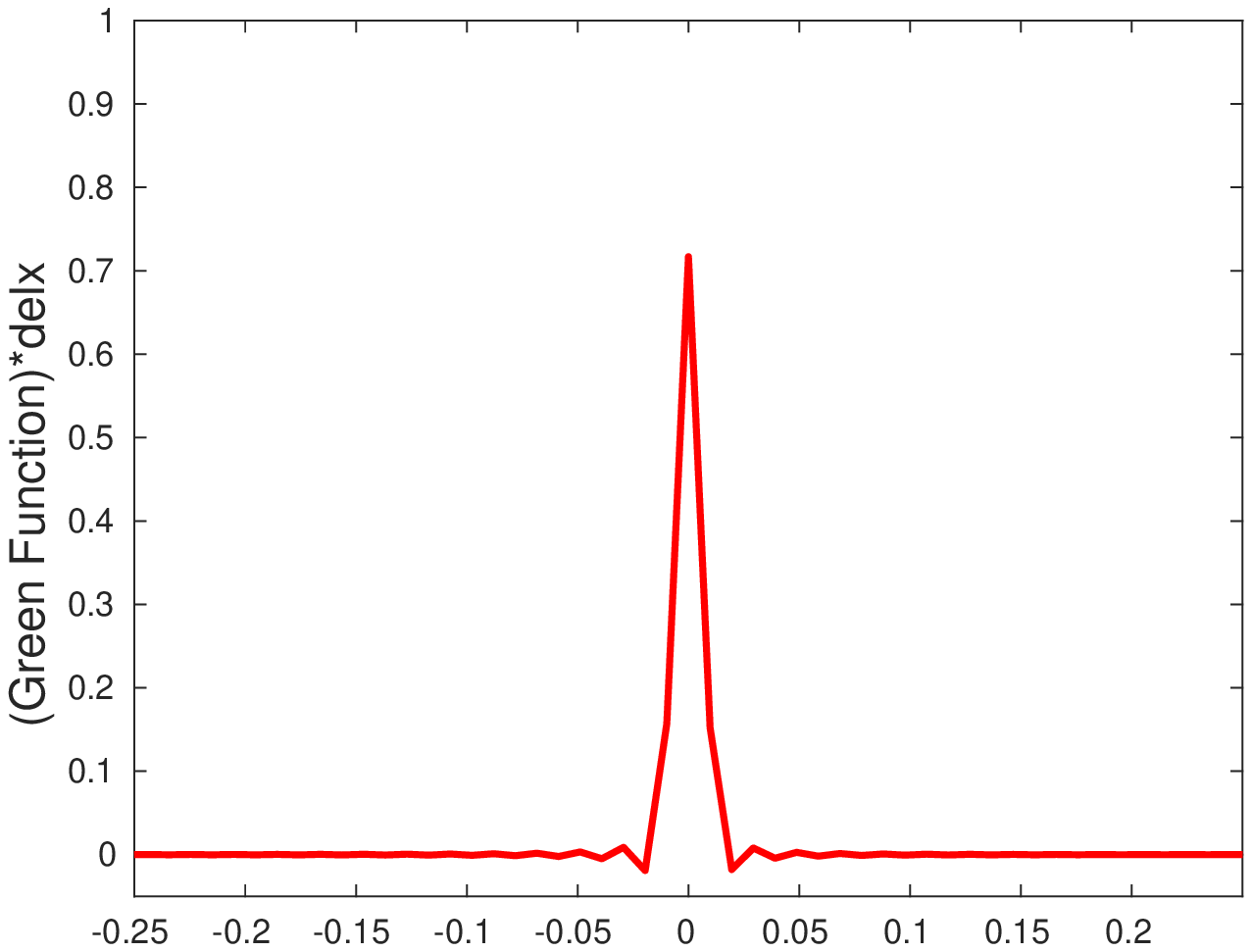}
    \caption{FST/CONV Green's function, truncated Fourier series.  Scaled
            by $\Delta x$. \label{FST_green_fig}}
  \end{subfigure}
  \begin{subfigure}[b]{.45\linewidth}
    \centering
     \includegraphics[width=3in]{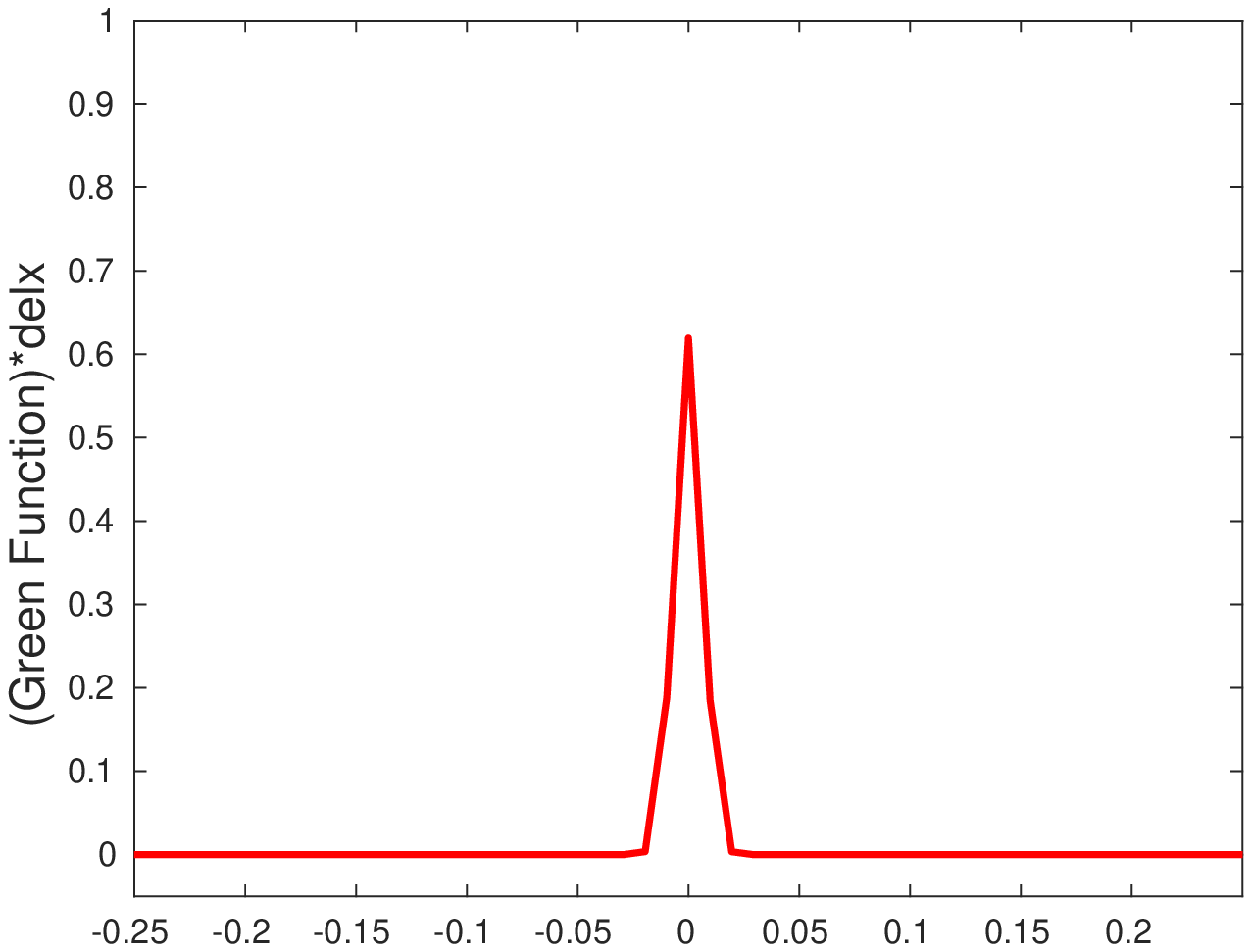}
    \caption{Monotone method, Green's function projected on linear basis functions.
             \label{Mono_green_fig}}
  \end{subfigure}
  }
 \caption{
European call option test: parameters in
Table \ref{simple_euro} but $T=.001$.  
FST/CONV method truncated Fourier series ($N=2048$).
Monotone method shows $\widetilde{g}( x, \Delta \tau, \alpha) \Delta x$, 
with $N=2048$, $\alpha = 4$.
\label{Green_figs}
}
\end{figure}

\begin{figure}[tbh]
\begin{center}
 \includegraphics[width=5.0in]{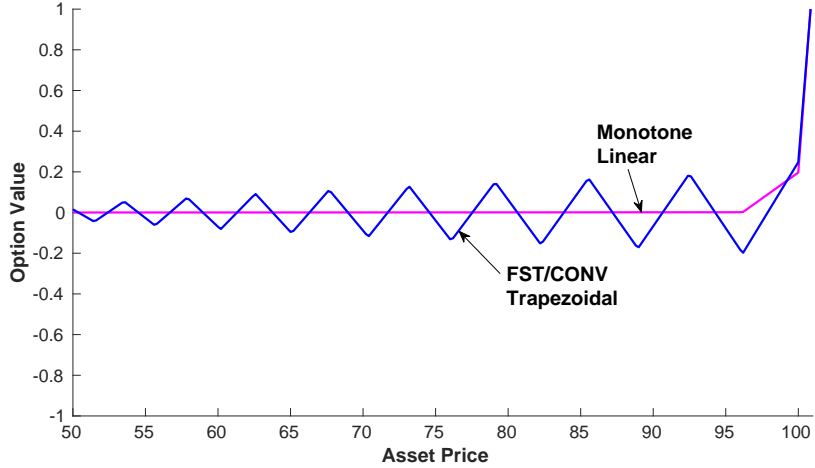}
\end{center}
\caption{European call option test.  Parameters in
             Table \ref{simple_euro} but $T=.001$.  $N=512$.
             For the monotone solution, $\alpha = 4$ (see
             Algorithm \ref{alg_start}).}
\label{Wiggle_plt}
\end{figure}

{\color{black}
\begin{remark}[Error in approximating equation (\ref{mono_5}) using equation (\ref{setup_1})]
An estimate of the error in computing the projected Green's function
is given in Appendix \ref{projection_convergence}, equation (\ref{green_error_2}).
We can see that a very small timestep effects the exponent in equation (\ref{green_error_2}).
For the extreme case of $T=.001$, $N=2056$, problem
in Table \ref{simple_euro},  we observe that for $\alpha = 8$,
then $test_2$ in Algorithm \ref{alg_start} is approximately $10^{-12}$,
indicating a very high accuracy projection can be achieved under extreme situations.
For the same problem ($2056$ nodes) with $T=.25$, we find that $test_2$ in 
Algorithm \ref{alg_start} is approximately $10^{-16}$ for $\alpha = 2$.
\end{remark}}

From these tests we can conclude both that the monotone method is 
robust for all timestep sizes and that for smooth problems and 
large timesteps, the monotone  method exhibits the expected
slower rate of convergence compared to high order techniques.

\subsection{Bermudan option with non-proportional discrete dividends}

Let us now assume that we have the same underlying process (\ref{jump_process}) as in the previous subsection, except
that the density function for $y = \log( \xi)$ is assumed normal
\begin{eqnarray}
    f(y) & = & \frac{1}{ \sqrt{2 \pi}~ \gamma} e^{ - \frac{ (y-\nu)^2}{2 \gamma^2} } 
              \label{merton_jumps}
\end{eqnarray}
with expectation $E[\xi] = e^{\nu+\gamma^2/2}$. Rather than a European option,  we will now consider a 
Bermudan put option which can be early exercised at fixed monitoring times $\tau_n$.
In addition, the underlying asset pays a fixed dividend amount $D$ at $\tau_n^-$, that is, immediately
after the early exercise opportunity in forward time. Between monitoring dates, the option price is 
given by equation (\ref{pde_1}). At monitoring dates we have the condition
\begin{eqnarray}
   v( x, \tau_n^+) & = & \max( v( \log(  \max(e^x -D, e^{x_{\min}}) ), \tau_n^-) , P( x ) )  \mbox{ with } \nonumber \\
                   P(x) & = &  {\mbox{ payoff }} = \max( K - e^x, 0)~.
              \label{control_test_1}
\end{eqnarray}
The expression $\max(e^x -D, e^{x_{\min}}) $ in equation (\ref{control_test_1})
ensures that the no-arbitrage condition  holds, that is, the dividend 
cannot be larger than the stock price, taking into account the localized grid.
Linear interpolation is used to evaluate the option value  in equation (\ref{control_test_1}).
The parameters for this problem are listed in Table \ref{bermudan} with the numerical results given in Table \ref{bermudan_results}.  All methods perform similarly, with second order convergence.
We can see here that once we use a linear interpolation to impose
the control, there is no benefit, in terms of convergence order, to using a high order method.

\begin{table}[tb]
  \begin{center}
    \begin{tabular}{ll} \toprule
      Expiry time & 10 years \\
      Strike K  & 100\\
      Payoff    & put \\
      Initial asset price $S_0$  & 100\\
      Risk-free rate $r$ & .05  \\
      Volatility $\sigma$  & .15 \\
      Dividend $D$ & 1.00 \\
      Monitoring frequency $\Delta \tau$ & 1.0 years\\
       $\lambda$  & .1 \\
      $\nu$ &  -1.08 \\
      $\gamma$ &  .4\\
      $x_{\max}$ & $ \log(S_0) + 10$\\
      $x_{\min}$ & $ \log(S_0) - 10$\\
      $\epsilon_1, \epsilon_2$ & $10^{-6}$ \\
      Asymptotic form $x \rightarrow \infty$  &$ A(x) = 0 $\\
           \bottomrule
    \end{tabular}
    \caption{Bermudan put option test.
     \label{bermudan}
      }
  \end{center}
\end{table}

\begin{table}[tb]
  \begin{center}
    \begin{tabular}{|l|l|l|l|l|l|l|l|l|} \hline
                & \multicolumn{4}{c|}{Monotone Methods} & \multicolumn{4}{c|}{FST/CONV}\\ \hline
                & \multicolumn{2}{c|}{Piecewise linear} &  \multicolumn{2}{c|}{Piecewise constant}
                & \multicolumn{2}{c|}{Trapezoidal} & \multicolumn{2}{c|}{Simpson} \\
                         \hline
    $N$ & Value & Ratio & Value & Ratio & Value & Ratio & Value & Ratio\\
                          \hline
    $2^9 $        & 24.811127744  &     & 24.806532754   &      & 24.801967268  &       & 24.802639420   &  \\
    $2^{10}$      & 24.789931363 &      &24.788800257  &        & 24.787670043  &       & 24.787731820  &   \\
    $2^{11}$      & 24.782264461 & 2.8  &2.4.781982815   & 2.6  & 24.781701225  &  2.4  & 24.781787212  & 2.5    \\
    $2^{12}$      & 24.781134292 & 6.8  &24.781063962   & 7.4   & 24.780993635  &  8.4  & 24.781007785  & 7.6   \\
    $2^{13}$      & 24.780822977 & 3.6  &24.780805394   & 3.6   & 24.780787811  &  3.4  & 24.780788678  & 3.6   \\
    $2^{14}$      & 24.780744620 & 4.0  &24.780740225   & 4.0   & 24.780735831  & 4.0   & 24.780737159  & 4.3  \\
            \hline
    \end{tabular}
    \caption{Bermudan put option test: value at $\tau=0, S=S_0$.  Parameters in
             Table \ref{simple_euro}.  $N=$ number of nodes. Ratio is the
             ratio of successive changes.
     \label{bermudan_results}}
  \end{center}
\end{table}


\section{Multiperiod mean variance optimal asset allocation problem}\label{sec8}

In this section we give an example of a realistic problem with complex controls, the 
{\em multiperiod mean variance optimal asset allocation problem}. Here we consider the case of an investor with a portfolio consisting
of a bond index and a stock index.  The {\em amount} invested in the stock index follows the process under the objective measure
\begin{eqnarray}
   \frac{dS_t}{S_{t^-}} &= (\mu -\lambda \kappa) dt + \sigma  dZ + d
     \biggl( \displaystyle \sum_{i=1}^{\pi_t} (\xi_i-1) \biggr)
    \label{jump_process_real}
\end{eqnarray}
with the double exponential jump size distribution (\ref{kou_jumps}), while the amount in the bond index follows
\begin{eqnarray}
  dB_t = rB_t~dt \label{bond_process}~.
\end{eqnarray}

The investor injects cash $q_n$ at time
time $t_n  \in \hat{\mathcal{T}}$ with total wealth at time $t$ being $W_t = S_t + B_t$. 
Let $ W_n^- = S_n^- + B_n^-$ be the total wealth before cash
injection.  
It turns out that in the multiperiod mean variance case,
in some circumstances, it is optimal to withdraw
cash from the portfolio \citep{Cui_2014,Dang2015a}. 
Denote this optimal
cash withdrawal as $c^*_n $.   The total
wealth after cash injection and withdrawal is then
\begin{eqnarray}
    W_n^+ & = & W_n^- + q_n -c^*_n~.
\end{eqnarray}
We then select an amount $b^*_n$ to invest in the bond,
so that
\begin{eqnarray}
   B_n^+ & = & b^*_n  ~~\mbox{ and } ~~
   S_n^+ = W_n^+ - b_n^* ~.
\end{eqnarray}
Since only cash withdrawals are allowed  we have $c^*_n \geq 0$.
The control at rebalancing time $t_n$ consists
of the pair $(b_n^*,~c^*_n)$. That is, after withdrawing
$c^*_n$ from the portfolio we rebalance to a portfolio with $S_n^+$
in stock and $B_n^+ $ in bonds. A no-leverage and no-shorting constraint is enforced by
\begin{eqnarray}
   0 \leq  b_n^*  \leq  W_n^+  ~.
\end{eqnarray}
In order to determine the mean-variance optimal solution to this asset allocation
problem, we make use of
the embedding result \citep{li-ng:2000,zhou-li:2000}.  
The mean-variance optimal strategy can be posed as
\begin{eqnarray}
    \min_{ \{ (b_0^*,c^*_0), \ldots, (b^*_{M-1}, c^*_{M-1}) \}}  & & E[ (W^* - W_T)^2] \nonumber \\
   {\mbox{ subject to}} & & \begin{cases}
                       (S_t, B_t) {\mbox{ follow processes }} (\ref{jump_process_real}),  
                            (\ref{bond_process})~~; ~~t \notin \hat{\mathcal{T}}\\
                          W_{n}^+ = S_{n}^- +B_n^- + q_n -c^*_n, \\
                          ~~~~~~~~~~S_n^+ =  W_n^+ -b_n^* , 
                                     B_n^+ = b_n^*
                       ~~;~~ t \in \hat{\mathcal{T}}\\
                       0 \leq b_n^*   \leq W_n^+ \\
                       c^*_n  \geq 0
                   \end{cases} ~,
             \label{embedded_problem}
\end{eqnarray}
where $W^*$ can viewed as a parameter which traces out the efficient frontier.

Let
\begin{eqnarray}
    Q_\ell = \sum_{j=\ell+1}^{M-1} e^{ -r( t_j- t_\ell)} q_j
                 \label{Q_eqn}
\end{eqnarray}
be the discounted future contributions to the portfolio at time $t_\ell$.
If
\begin{eqnarray}
   (W_n^- + q_n)   > W^* e^{-r(T-t_n)} - Q_{n}~,
\end{eqnarray}
then the optimal strategy is to
withdraw cash $c_n^* =  W_n^- + q_n - (W^* e^{-r(T-t_n)} - Q_{n})$ from the portfolio,
and invest the remainder $\bigl( W^* e^{-r(T-t_n)} - Q_{n} \bigr) $  in the risk free  asset.
This is optimal in this case since then $E[ (W^* - W_T)^2] = 0$
\citep{cui-li-wang-zhu:2012,Dang2015a}, which is the minimum of problem (\ref{embedded_problem}).

In the following we will refer to any cash withdrawn from the portfolio as a {\em surplus} or {\em free cash flow} \citep{Bauerle2015}.  For the sake of discussion, we will assume that the surplus
cash is invested in a risk-free asset, but does not contribute to the
computation of the terminal mean and variance.  Other possibilities
are discussed in \citet{Dang2015a}.

The solution of
problem (\ref{embedded_problem}) is the so-called {\em pre-commitment} solution. 
We can interpret the pre-commitment solution
in the following way.  At $t=0$, we  decide which Pareto point is
desirable (that is, a point on the efficient frontier).  This fixes
the value of $W^*$.
At  any time $t>0$,
we can  regard the optimal policy as the time-consistent solution to the
problem of minimizing the expected quadratic loss with respect to the
fixed target wealth $W^*$, which can be viewed as a useful
practical objective function \citep{vigna:2014,Vigna_2017b}.

\subsection{Optimal control problem}

A brief overview of the PIDE for the solution of the mean-variance optimal control problem is given below (we refer the
reader to \citet{dang-forsyth:2014a} for additional details).
 
Let the value function $v(x,b,\tau)$ with $\tau = T-t$ be defined as
\begin{eqnarray}
   v(x, b, \tau) & = & \inf_{ \{ (b^*_0,c^*_0), \ldots, (b^*_{M}, c^*_{M}) \}} 
                      \biggl \{
                           E \biggl[ ~( \min( W_T - W^*, 0) )^2  ~\bigg \vert ~ \log S(t) = x, B(t) = b
                             \biggr]
                      \biggr\} ~. 
           \label{mean_var_embed_a}
\end{eqnarray}
Let the set of observation times  backward in time be $ \mathcal{T} = \{ \tau_0, \tau_1, \ldots, \tau_M\}$.
For $\tau \notin \mathcal{T}$, $v$ satisfies
\begin{eqnarray}
  v_\tau & = & \mathcal{L}v + rb v_b ~~~\mbox{ where }  \nonumber \\
          \mathcal{L} v & \equiv&  \frac{\sigma^2}{2} v_{xx} + (\mu - \frac{\sigma^2}{2} - \lambda \kappa) v_x
                  - (\mu +\lambda) v + \lambda \int_{-\infty}^{\infty} v( x + y) f(y) ~dy \nonumber \\
          v( x, b, 0) & = & ( \min( e^x + b -W^*, 0) )^2 ~ \label{mean_var_pide}
\end{eqnarray}
on the localized domain $(x,b) \in [x_{\min}, x_{\max} ] \times [0, b_{\max}]$.

If $g(x, \tau)$ is the Green's function of $v_{\tau} = \mathcal{L} v$ then the solution of 
equation (\ref{mean_var_pide}) at $\tau_{n+1}^-$, given the solution at $\tau_n^+$,
$\tau_n \in \mathcal{T}$ is
\begin{eqnarray}
    v(  x, b, \tau_{n+1}^-) & = & \int_{x_{\min} }^{x_{\max}} g(x - x^{\prime}, \Delta \tau)
                                   v( x^{\prime}, b e^{rb \Delta \tau}, \tau_n^+) ~~~\mbox{ with }~~~ 
                             \Delta \tau = \tau_{n+1} - \tau_n ~. \label{mean_var_1}
\end{eqnarray}
Equation (\ref{mean_var_1}) can be regarded as a combination of a
Green's function step for the PIDE $v_{\tau} = \mathcal{L} v$ and a characteristic
technique to handle the $rbv_b$ term.
At rebalancing times $\tau_n \in \mathcal{T}$, 
\begin{eqnarray}
     v( x, b, \tau_n^+) & = & \min_{ (b^*, c^*)}   v( x^{\prime}, b^*, \tau_n^-) \nonumber \\
    {\mbox{ subject to}} & & \begin{cases}
                           c^* = \max( e^x + b + q_{M-n} -Q_{M-n}, 0) \\
                           W^{\prime} = e^x + b + q_{M-n} - c^*\\
                           0 \leq b^* \leq W^{\prime} \\
                           x^{\prime} = \log\bigl( \max( W^{\prime} - b^*, e^{x_{\min}} )
                                            \bigr)
                   \end{cases} ~
             \label{mean_var_control}
\end{eqnarray}
where $Q_{\ell}$ is defined in equation (\ref{Q_eqn}).

\subsection{Computational details}
We solve problem (\ref{mean_var_embed_a}) combined with the optimal control 
(\ref{mean_var_control}) on the localized domain 
 $(x,b) \in [x_{\min}, x_{\max} ] \times [0, b_{\max}]$.
We discretize in the $x$ direction using an equally
spaced grid with $N_x$ nodes and an unequally
spaced grid in the $B$ direction with $N_b$ nodes.
Set $B_{\max} = e^{x_{\max}}$ and denote the discrete solution at $(x_m, b_j, \tau_n^+)$ by
\begin{eqnarray}
   (v_{m,j}^n)^+ & = & v( x_m, b_j, \tau_n^+)
         \nonumber \\
    (v^n)^+ & = & \{ (v_{m,j}^n)^+ \}_{m=-N_x/2,\ldots,N_x/2-1; j=1, \ldots, N_b} \nonumber
                \\
    (v^n_j)^+ & = & [ (v_{-N_x/2,j}^n)^+, \ldots, (v_{N_x/2-1,j}^n)^+ ] .
\end{eqnarray}

Let $\mathcal{I}_{ \Delta x, \Delta b}( x, b) (v^n)^- $ be a two dimensional linear
interpolation operator acting on the discrete solution values $(v^n)^-$.
Given the solution at $\tau_n^+$, we use Algorithm \ref{alg_monotone}
to advance the solution to $\tau_{n+1}^-$.  For the mean
variance problem, we extend this algorithm to approximate
equation (\ref{mean_var_1}),  which is described in
Algorithm \ref{advance_time}.

\begin{algorithm}[!h]
\caption{
\label{advance_time}
Advance time $(v^n)^+ \rightarrow (v^{n+1})^-$.
}
 \begin{algorithmic}[1]
    \REQUIRE $ (v^n)^+~;~ \widetilde{G} = \{\widetilde{G}( \omega_m, \Delta \tau, \alpha) \}, 
                 ~ m=-N_x/2, \ldots, N_x/2-1$
               (from Algorithm \ref{alg_start}) 
    \FOR[Advance time loop]{$j=1,\ldots, N_b$}
       \STATE $v^{int}_{m,j} = \mathcal{I}_{ \Delta x, \Delta b} (x_m, b_j e^{r \Delta \tau}) (v^n)^+ ~;~
                                  m=-N_x/2, \ldots, N_x/2-1$
         \STATE $ \widetilde{V} = FFT [~ v^{int}_{j} ~] $
         \STATE $(v^{n+1}_j)^- =  iFFT[~ \widetilde{V} \circ \widetilde{G} ~]~~~~~~~~~~~~~~~~~~~~~~~~~~~~~~~~~~~~~~~~~~~~~~~~~~~$ \COMMENT{iFFT( Hadamard product )}
    \ENDFOR \COMMENT{End advance time loop}

 \end{algorithmic}
\end{algorithm}

In order to advance the solution from $\tau_{n+1}^-$ to $\tau_{n+1}^+$,
we approximate the solution to the optimal control 
problem (\ref{mean_var_control}). 
The optimal control is approximated by discretizing the candidate control $b^*$
using the discretized $b$ grid and exhaustive search:
\begin{eqnarray}
     v( x_m, b_j, \tau_n^+) & = & \min_{ (b^*, c^*)}   
        \mathcal{I}_{ \Delta x, \Delta b} (( x^{*}, b^*)  ( v^{n+1})^- \nonumber \\
    {\mbox{ subject to}} & & \begin{cases}
                           c^* = \max( e^{x_m} + b_j + q_{M-n} -Q_{M-n}, 0) \\
                           W^{\prime} = e^{x_m} + b_j + q_{M-n} - c^*\\
                           b^* \in \{ b_1, \ldots, \min( b_{\max}, W^{\prime}) \} \\
                           x^{*} = \log\bigl( 
                                     \max( W^{\prime} - b^*, e^{x_{\min}} ) 
                                       \bigr)
                   \end{cases} ~.
             \label{mean_var_control_b}
\end{eqnarray}
This is a convergent algorithm to the solution of the original control problem as $N_x, N_b \rightarrow \infty$. 
This can be proved using similar steps as in the finite
difference case \citep{dang-forsyth:2014a}. For brevity we omit the proof. 

Using the control determined from solving problem (\ref{mean_var_embed_a}),
we can determine $E[W_T]$ and $std[W_T]$ by solving an additional
linear PIDE, see \citep{dang-forsyth:2014a} for details.


\begin{remark}[Practical implementation enhancements]
As noted by several authors, since the Green's function  and the solution is real,
the Fourier coefficients satisfy symmetry relations. Hence 
$\widetilde{G}( \omega_k, \Delta \tau, \alpha)$ and $ \widetilde{V}$ need to be computed
and stored only for $\omega_k \geq 0$.  
It is also possible to arrange the 
step in line $2$ of Algorithm $4$ and the 
optimal control step of (\ref{mean_var_control_b}) so that only a single interpolation error is introduced at each node.
Note that the Fourier series representation of the Green's function
is only used to compute the projection of the Green's function
onto linear basis functions.  After this initial step, we use FFTs
only to efficiently carry out a dense matrix-vector multiply (the convolution)
at each step.  Use of the FFT here is algebraically identical to carrying out the convolution
in the physical space.  The only approximation being used in this step is the periodic extension of the solution.
\end{remark}

\subsection{Numerical example}

The data for this problem is given in Table \ref{mean_var_num_1}.
The data was determined by fitting to the monthly returns
from the Center for Research in Security Prices (CRSP) through
Wharton Research Data Services, for the period {\em{1926:1- 2015:12}}.\footnote{More
specifically, results presented here were calculated based
on data from Historical Indexes, \copyright 2015 Center for Research
in Security Prices (CRSP), The University of Chicago Booth School of
Business.  Wharton Research Data Services was used in preparing this article.
This service and the data available thereon constitute valuable intellectual
property and trade secrets of WRDS and/or its third-party suppliers.}
We
use the monthly CRSP value-weighted (capitalization weighted) total
return index (``vwretd''), which includes
all distributions for all domestic stocks trading on major US exchanges, and the monthly 90-day Treasury bill return index from CRSP.  Both this index
and the equity index are in nominal terms, so we adjust them for
inflation by using the US CPI index (also supplied by CRSP).
We use real indexes since investors saving for retirement
are focused on real
(not nominal) wealth goals.

\begin{table}[tb]
  \begin{center}
    \begin{tabular}{ll} \toprule
      Expiry time $T$ & 30 years \\
      Initial wealth   & 0 \\
      Rebalancing frequency & yearly\\
      Cash injection $\{q_i\}_{i=0, \ldots, 29}$ & 10\\
      Real interest rate $r$ & .00827  \\
      Volatility $\sigma$  &  .14777 \\
      $ \mu $   &  .08885 \\
       $\lambda$  & .3222 \\
      $\eta_1$ &  4.4273 \\
      $\eta_2$ &   5.262\\
      $p_{u}$ &  0.2758\\
      $x_{\max}$ & $ \log(100) + 5 $\\
      $x_{\min}$ & $ \log( 100) - 10$\\
      $\epsilon_1, \epsilon_2$ & $10^{-6}$ \\
      Asymptotic form $E[ (W_T - W^*)^2], x \rightarrow \infty$ &
                        $ A(x) = 0$ \nonumber \\
           \bottomrule
    \end{tabular}
    \caption{Multiperiod mean variance example.  
    Parameters determined by fitting to the real (inflation adjusted) CRSP data for the
    period {\em 1926:1-2015:12}.  Interest rate is the average real return on 90 day T-bills.
     \label{mean_var_num_1}}
  \end{center}
\end{table}

As a first test, we fix $W^* = 1022$, and then increase the number nodes
in the $x$ direction ($N_x$) and in the $b$ direction ($N_b$). We use the
monotone scheme, with linear basis functions.  In Table~\ref{mean_var_table_a}, we show the value
function $v( 0, 0,T)$ and the mean $E[W_T]$  and standard deviation $std[W_T]$ of the final wealth,
which are of practical importance.
The value function shows smooth second order convergence, which is to be expected.
Even though the optimal control is correct only to order $\Delta b$ (since
we optimize by discretizing the controls and using exhaustive search),
the value function is correct to $O(\Delta b)^2$ (since it is an extreme
point).

We expect that the  derived quantities $E[W_T], std[W_T]$, which are based on the  controls computed as a 
byproduct of computing the value function, should show a lower order convergence. Recall that these quantities 
are evaluated by storing the controls and then solving a linear PIDE.  In fact we do see  somewhat erratic 
convergence for these quantities.
As an independent check, we used the stored controls from solving for the value function (on
the finest grid),
and then carried out Monte Carlo simulations to directly compute the mean and
standard deviation of the final wealth.  The results are shown in Table \ref{MC_table}.

\begin{table}[tb]
  \begin{center}
    \begin{tabular}{|l|l|l|l|l|l|l|l|} \hline
    $N_x$ & $N_b$ & Value function & Ratio & $E[W_T]$ & Ratio & $std[W_T]$ & Ratio \\
             \hline
     512  & 305 &  97148.899100 & N/A   & 824.02599269 & N/A & 240.73884508 & N/A \\
     1024 & 609 &  97042.740997 & N/A   & 824.07104985 & N/A & 240.55534019 & N/A \\
     2048 & 1217 & 97014.471301 & 3.8  & 824.09034690 &  2.3 & 240.51245396 &  4.3 \\
     4096 & 2433 & 97007.286530 & 3.9  & 824.08961667 & -26  & 240.49691620 & 2.7 \\
     8192 & 4865 & 97005.451814 & 3.9  & 824.09295889 & -.22 & 240.49585213 & 14.6 \\
            \hline
    \end{tabular}
    \caption{Test of convergence of optimal multiperiod mean variance investment strategy.  
              Monotone method, linear basis functions.
             Parameters in Table \ref{mean_var_num_1}.  Fixed $W^* = 1022$.
              Ratio is the
             ratio of successive changes.
     \label{mean_var_table_a}}
  \end{center}
\end{table}

\begin{table}[tb]
  \begin{center}
    \begin{tabular}{|l|l|l|} \hline
    $N_{sim}$ & $E[W_T]$  & $std[W_T]$ \\
     \hline
    $1.6 \times 10^5$ & 824.3425 (1.55) & 240.2263 \\
    $6.4 \times 10^5$ & 823.6719 (0.78) & 240.7278 \\
    $2.56 \times 10^6 $ &824.0077 (0.39) &240.4336 \\
    $1.024 \times 10^7$ & 824.1043 (0.19) & 240.5217 \\
            \hline
    \end{tabular}
    \caption{Monte Carlo simulation results, based on optimal controls from solving for the
             value function using the monotone Fourier technique.  Numbers in brackets are
             the standard error, 99\% confidence level, for the mean.
             Compare with Table \ref{mean_var_table_a}.
             Parameters in Table \ref{mean_var_num_1}.  Fixed $W^* = 1022$.
     \label{MC_table}}
  \end{center}
\end{table}

Of more practical interest is the following computation.  In Table \ref{const_wt}
we show the results obtained by rebalancing to a constant weight in equities
at each monitoring date.  We specify that the portfolio is rebalanced to
$.60$ in stocks and $.40$ in bonds (a common default recommendation).
We then solve for the value function using the monotone Fourier method, 
allowing $W^*$ to vary, but fixing the expected value so that $E[W_T]$ is the
same as for the $60:40$ constant proportion strategy.  This is done by using
a Newton iteration, where each evaluation of the residual function requires
a solve for the value function and the expected value equation.
The results of this test are shown in Table \ref{mean_specd}.
In this case, fixing the mean and allowing $W^*$ to vary, results
in smooth convergence of the standard deviation.
From a practical point of view, we can see that the optimal
strategy has the same expected value as the constant
proportion strategy, but the standard deviation is reduced from
$512$ to $241$, and the median of the optimal strategy
is $936$ compared to a median of $704$ for the constant
proportion strategy.
A heat map of the optimal strategy is shown in Figure \ref{optimal_stock}.

\begin{table}[tb]
  \begin{center}
    \begin{tabular}{|l|l|l|} \hline
    $E[W_T]$  & $std[W_T]$ & $Median[W_T]$ 
            \\
     \hline
    824.10047  & 511.8482 & 704 \\
            \hline
    \end{tabular}
    \caption{Portfolio rebalanced to $.60$ in stocks and $.40$ in bonds at
             each monitoring date.  Closed form expression for mean and standard
             deviation.  Median computed using Monte Carlo simulation.
             Parameters in Table \ref{mean_var_num_1}. 
     \label{const_wt}}
  \end{center}
\end{table}

\begin{table}[tb]
  \begin{center}
    \begin{tabular}{|l|l|l|l|l|l|} \hline
    $N_x$ & $N_b$ &$E[W_T]$  & $std[W_T]$ & Ratio  
          \\
     \hline
   512  & 305    & 824.10047 &  240.79440842 &  N/A  \\
   1024 & 609 &    824.10047 &  240.57925928 & N/A    \\
   2048 & 1217  &  824.10047 &  240.52022512  & 3.6   \\
   4096 & 2433  &   824.10047 & 240.50571976 & 4.1   \\
   8192 & 4865 &    824.10047 & 240.50220544 & 4.1  \\
            \hline
    \end{tabular}
    \caption{At each refinement level $W^*$ is determined so that $E[W_T] = 824.10047$.
             The median on the finest grid is computed by storing the controls and using
             Monte Carlo simulation.  $Median[W_T] = 936$. Ratio is the ratio of
             successive changes.
             Parameters in Table \ref{mean_var_num_1}.  
     \label{mean_specd}}
  \end{center}
\end{table}

\begin{figure}[tb]
\centerline{
     \includegraphics[width=4.in]{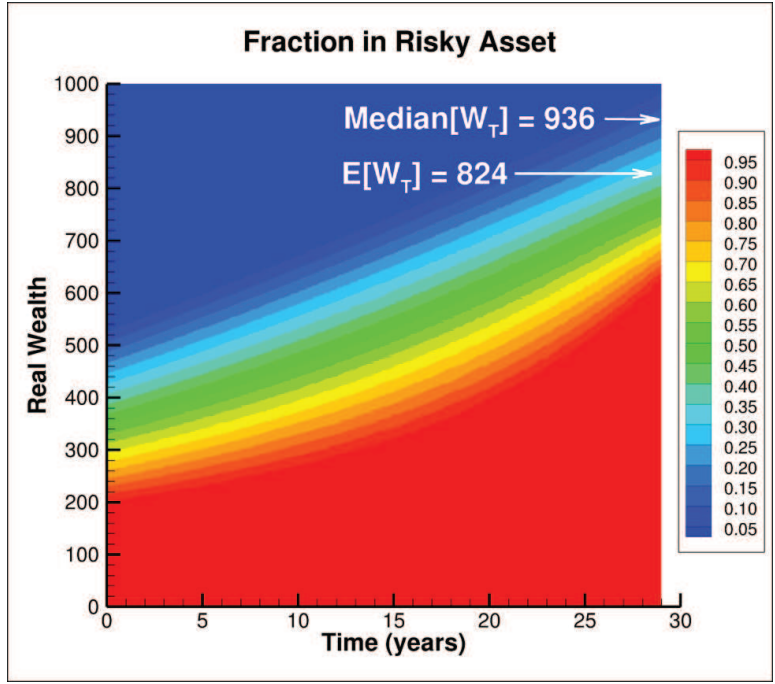}
  }
 \caption{Optimal strategy, fraction of portfolio invested
in stock, as a function of current total real wealth $W_t = S_t + B_t$
and forward time $t$. Parameters in Table \ref{mean_var_num_1}.
\label{optimal_stock}
}
\end{figure}

\section{Conclusions}\label{sec9}

Many problems in finance give rise to discretely monitored complex control problems.  
In many cases, the optimal controls are not of a simple  bang-bang type. It then happens that  
a numerical procedure must be used to determine the optimal control at discrete points in the  
physical domain.  In these situations, there is little hope of obtaining a high order accurate
solution, after the control is applied.  If we desire a monotone scheme, which increases robustness
and reliability for our computations, then we are limited to the use of linear interpolation, hence 
we can get at most second order accuracy.

{\color{black}
Traditional FST/CONV methods assume knowledge of the Fourier
transform of the Green's function but then approximate this 
function by a truncated Fourier series. As a result these methods are not monotone. Instead when the 
Fourier transform of the Green's 
function is known, then we carry out a pre-processing step by projecting  the Green's function
(in the physical space) onto a set of 
linear basis functions. These integrals can then be computed to within a 
specified tolerance and this allows 
us to guarantee a monotone scheme to within the tolerance. 
This monotone scheme is robust to small timesteps, 
which is observably not the case for the standard FST/CONV methods, 
and indeed is a major pitfall of the latter methods.}

When the Green's function depends on time only through the timestep size and the monitoring dates for the control are equally spaced (which is typically the case), then the final monotone algorithm  has the same complexity per step as the original 
FST/CONV algorithms, and the same order of convergence for 
smooth control problems. It is a simple process to add this preprocessing
step to existing FST/CONV software. This results in more robust and more  reliable algorithms for optimal stochastic control problems.

\section{Acknowledgements and Declaration of Interest}
This work was supported by the Natural Sciences and
Engineering Research Council of Canada (NSERC).  The authors report no conflicts of interest. 
The authors alone are responsible for the content and writing of the paper.


\appendix

\section*{Appendices}

\appendix

\section{Green's functions}\label{green_appendix}

Consider the PIDE
\begin{eqnarray}
    v_{\tau} & = & \frac{\sigma^2}{2} v_{xx} + (\mu - \frac{\sigma^2}{2} - \lambda \kappa) v_x
                   -( \rho + \lambda) v + \lambda \int_{-\infty}^{+\infty} v( x + y) f(y) ~dy
    ~. \label{green_app_1}
\end{eqnarray}
If,  for example, $\rho = \mu = r$ where $r$ is the risk-free rate, then this
is the option pricing equation, while if $\rho = 0$ then the right hand side
of equation (\ref{green_app_1}) is $\mathcal{L} v$ in the mean variance case.

Let 
\begin{eqnarray}
   v( \omega, \tau ) & = & \int_{-\infty}^{\infty} V(\omega, \tau) e^{2 \pi i \omega x}~d\omega \nonumber \\
   f( \omega, \tau ) & = & \int_{-\infty}^{\infty} F( \omega, \tau) e^{2 \pi i \omega x}~d\omega
  ~. \label{green_app_2}
\end{eqnarray}
Substituting equation (\ref{green_app_2}) into equation (\ref{green_app_1}) gives
\begin{eqnarray}
   V(\omega, \tau)_{\tau} & = & \Psi(\omega) V( \omega, \tau ) ~~ \mbox{ where } \label{green_app_3a} \\
   \Psi(\omega) & = & \biggl( -\frac{\sigma^2}{2} (2 \pi \omega)^2 
                       + ( \mu - \lambda \kappa - \frac{\sigma^2}{2}) ( 2 \pi i \omega)
                        -( \rho + \lambda) +  \lambda \overline{F}( \omega) 
                   \biggr) ~, \nonumber 
\end{eqnarray}
with $\overline{F}( \omega)$ being the complex conjugate of ${F}( \omega)$.
Integrating equation (\ref{green_app_3a}) gives
\begin{eqnarray}
    V(\omega, \tau + \Delta \tau) = e^{ \Psi(\omega) \Delta \tau}  V(\omega, \tau) ~, \nonumber
\end{eqnarray}
from which we can deduce that the Fourier transform of the Green's function $G( \omega, \Delta \tau)$
is
\begin{eqnarray}
    G( \omega, \Delta \tau) & = & e^{ \Psi(\omega) \Delta \tau} ~.
              \label{green_app_4}
\end{eqnarray}

In the case of a double exponential jump distribution (\ref{kou_jumps}), then
\begin{eqnarray}
  \overline{F}( \omega) & = & \frac{ p_{u}}{ 1 - 2 \pi i \omega/ \eta_1} 
                              + \frac{1 - p_{u}}{ 1 + 2 \pi i \omega / \eta_2} \nonumber 
\end{eqnarray}
while in the case of a log-normal jump size distribution (\ref{merton_jumps})
\begin{eqnarray}
   \overline{F}( \omega) & = & e^{ 2( \pi i \omega \nu - ( \pi \omega \gamma)^2 ) } ~.
                       \nonumber 
\end{eqnarray}
From equation (\ref{green_app_3a}) and (\ref{green_app_4}) we have
\begin{eqnarray}
     G( 0 , \Delta \tau) & = & e^{- \rho \Delta \tau} ~, \nonumber
\end{eqnarray}
which means that  in these cases $C_1 = \int_{\mathbb{R}} g(x, \Delta \tau)~dx$ is
\begin{eqnarray}
    C_1 & = & \begin{cases}
                 e^{-r \Delta \tau} & {\mbox{ option pricing }}\\
                 1                  & {\mbox{ mean variance asset allocation}}
              \end{cases} ~. \nonumber
\end{eqnarray}

\section{Convergence of truncated Fourier series for the projected Green's functions.}

\label{projection_convergence}
Since the Green's function for equation (\ref{green_app_1}) is a smooth function
for any finite $\Delta \tau$, we can expect uniform convergence of the
Fourier series to the exact Green's function, assuming that
$\sigma > 0$.  This can also be seen from the 
exponential decay of the Fourier coefficients, which we
demonstrate in this Appendix.  Since the exact Green's function
is non-negative, the projected Green's function (\ref{g_tilde}) then
converges to a non-negative value at every point $y_j$.  Consider the
case of the truncated projection on linear basis functions
\begin{eqnarray}
    \widetilde{g}( y_j, \Delta \tau , \alpha) & = & 
                                           \frac{1}{P} \displaystyle \sum_{k=-\alpha N/2}^{\alpha N/2-1}
                                                     e^{ 2 \pi i \omega_k  y_j}
                                          \biggl( \frac{ \sin^2 \pi \omega_k \Delta x}
                                                         { ( \pi \omega_k \Delta x)^2}
                                           \biggr) G(\omega_k, \Delta \tau) ~~~  \mbox{ with } ~~~
                    \omega_k = \frac{k}{P} ~  \nonumber 
\end{eqnarray}
and $\Delta x = \frac{P}{N}$. The error in the truncated series is then
\begin{eqnarray}
   | \widetilde{g}( y_j, \Delta \tau , \alpha) - \widetilde{g}( y_j, \Delta \tau , \infty)|
    & = & \bigg \vert \displaystyle \frac{1}{P} \displaystyle \sum_{k=\alpha N/2}^{\infty}
                                                     e^{ 2 \pi i \omega_k  y_j}
                                          \biggl( \frac{ \sin^2 \pi \omega_k \Delta x}
                                                         { ( \pi \omega_k \Delta x)^2}
                                           \biggr) G(\omega_k, \Delta \tau) \nonumber \\
                            & & + 
                             \frac{1}{P} \displaystyle \sum_{k=-\infty}^{ - \alpha N/2-1}
                                                     e^{ 2 \pi i \omega_k  y_j}
                                          \biggl( \frac{ \sin^2 \pi \omega_k \Delta x}
                                                         { ( \pi \omega_k \Delta x)^2}
                                           \biggr) G(\omega_k, \Delta \tau)
                      \bigg \vert \nonumber \\
          & \leq & \displaystyle \frac{2}{P}
                  \displaystyle \sum_{k=\alpha N/2}^{ \infty}
                                           \frac{ 1}
                                                         { ( \pi \omega_k \Delta x)^2}
						     | G(\omega_k, \Delta \tau) | 
                   \nonumber \\
             &\leq & \displaystyle \frac{2}{P} \cdot \frac{ 4 }{ \pi^2 \alpha^2}
                   \sum_{k=\alpha N/2}^{\infty} | G(\omega_k, \Delta \tau) | ~.
          \label{bound_app_1}
\end{eqnarray}
Noting that $Re( \overline{F}( \omega) ) \leq 1 $, we then have
\begin{eqnarray}
  Re( \Psi(\omega) ) & = & - \frac{ \sigma^2 (2 \pi \omega)^2}{2} -(\rho + \lambda)
                           + \lambda Re( \overline{F}( \omega) ) \nonumber \\
                     & \leq & - \frac{ \sigma^2 (2 \pi \omega)^2}{2} -(\rho + \lambda)  + \lambda               
                      \leq   - \frac{ \sigma^2 (2 \pi \omega)^2}{2} \nonumber
\end{eqnarray}
since $\rho \geq 0$.
Hence
\begin{eqnarray}
    | G(\omega, \Delta \tau) | & = & | e^{  \Psi(\omega) \Delta \tau} | 
                                \leq  e^{ - \frac{ \sigma^2 (2 \pi \omega)^2 \Delta \tau}{2} }.
              \label{bound_app_3}
\end{eqnarray}
If we let $C_4 = \frac{ 2 \sigma^2  \pi^2 \Delta \tau}{P^2} $ then equations (\ref{bound_app_3}) and (\ref{bound_app_1}) implies
\begin{eqnarray}
  | \widetilde{g}( y_j, \Delta \tau , \alpha) - \widetilde{g}( y_j, \Delta \tau , \infty)|
  & \leq & \frac{8}{ P \pi^2 \alpha^2}
              \displaystyle \sum_{k=\alpha N/2}^{\infty}  e^{ -C_4 k^2} .
           \nonumber 
\end{eqnarray}
Bounding the sum gives
\begin{eqnarray}
   | \widetilde{g}( y_j, \Delta \tau , \alpha) - \widetilde{g}( y_j, \Delta \tau , \infty)| ~
      & \leq & ~\frac{ 8}{ P \pi^2 \alpha^2} \cdot \frac{ e^{- C_4 N^2 \alpha^2 /4}}{1 - e^{- C_4 N \alpha}}
     ~.
      \label{bound_app_5}
\end{eqnarray}

Consider the monotonicity test in Algorithm \ref{alg_start}, line \ref{test_1}, given by
\begin{eqnarray}
   test_1 = \displaystyle \sum_j \Delta x \min( \widetilde{g}( y_j, \Delta \tau, \alpha) , 0) . \nonumber
\end{eqnarray}
Noting that $\widetilde{g}( y_j, \Delta \tau , \infty) \geq 0$,
and from equation (\ref{bound_app_5}) and $ \sum_j \Delta x = P$, we  have
\begin{eqnarray}
   | test_1 | & \leq &  
                 \frac{ 8}{ \pi^2 \alpha^2} 
    \cdot  \frac{ e^{- C_4 N^2 \alpha^2 /4}}{1 - e^{- C_4 N \alpha}} ~,
            \label{test_1_bound}
\end{eqnarray}
so that usually this test is satisfied to within round off for $\alpha = 2, 4$.

Consider now the accuracy test on line \ref{test_2} of Algorithm \ref{alg_start}, given by
\begin{eqnarray}
   test_2 = \displaystyle \max_j \Delta x | \widetilde{g}( y_j, \Delta \tau, \alpha) -
                              \widetilde{g}( y_j, \Delta \tau, \alpha/2)| \nonumber
\end{eqnarray}
which we see is bounded by
\begin{eqnarray}
   |test_2| & \leq  & \Delta x \displaystyle \max_j 
                        \biggl(  | \widetilde{g}( y_j, \Delta \tau, \alpha) - 
                             \widetilde{g}( y_j, \Delta \tau, \infty)  |
                     + | \widetilde{g}( y_j, \Delta \tau, \alpha/2) - 
                             \widetilde{g}( y_j, \Delta \tau, \infty)  | 
                       \biggr) \nonumber \\
                 & \leq & \Delta x \displaystyle \max_j 
                          \biggl( 2 | \widetilde{g}( y_j, \Delta \tau, \alpha/2) - 
                             \widetilde{g}( y_j, \Delta \tau, \infty)  | \biggr)  \nonumber \\
                 & \leq &   
                     \frac{64}{ \pi^2 \alpha^2} \cdot \frac{ \Delta x}{P} 
                  \cdot \frac{ e^{ - C_4 N^2 \alpha^2 /16} }
                       { 1 - e^{- C_4 N \alpha/2}    } ~. \label{green_error_2}
\end{eqnarray}
This test will also be satisfied for small values of $\alpha$, although this will require larger values of 
$\alpha$ than for the monotonicity test (\ref{test_1_bound}).


\end{document}